\documentclass[12pt,a4paper]{amsart}
\usepackage[utf8]{inputenc}
\usepackage[a4paper,margin=2.6cm]{geometry}
\usepackage{amsmath,amssymb,amsthm,mathtools,abraces}
\usepackage{color,xcolor}
\usepackage[hypertexnames=false,hyperfootnotes=false,colorlinks=true,linkcolor=blue,%
citecolor=purple,filecolor=magenta,urlcolor=cyan,unicode,linktocpage=true,pagebackref=false]{hyperref}
\usepackage{yhmath}
\usepackage{verbatim}

\newcommand\be{\begin{equation}}
\newcommand\ee{\end{equation}}

\newcommand\p{\partial}

\newcommand\Tr{{\rm Tr}\,}
\newcommand\diag{{\rm diag}\,}
\newcommand{\<}{\left <}
\renewcommand{\>}{\right >}
\newcommand{\Normord}[1]{:\mathrel{#1}:}

\def\lvac{\left <0\right |}
\def\rvac{\left |0\right >}

\DeclareMathOperator{\res}{Res}
\DeclareMathOperator{\odd}{odd}
\DeclareMathOperator{\even}{even}
\DeclareMathOperator{\gl}{\mathfrak{gl}}
\DeclareMathOperator{\go}{\mathfrak{go}}

\DeclareMathOperator{\DP}{DP}
\DeclareMathOperator{\Gr}{Gr}
\DeclareMathOperator{\KdV}{KdV}
\DeclareMathOperator{\BKP}{BKP}
\DeclareMathOperator{\sppan}{span}

\def\normord{ {\scriptstyle  \genfrac{}{}{0pt}{2}{\bullet}{\bullet} }}

\newtheorem{theorem}{Theorem}
\newtheorem{lemma}{Lemma}[section]
\newtheorem{proposition}[lemma]{Proposition}
\newtheorem{corollary}[lemma]{Corollary}
\newtheorem{remark}{Remark}[section]

\newtheorem{conjecture}{Conjecture}[section]
\newtheorem*{theorem*}{Theorem}

\newtheorem{definition}{Definition}

\numberwithin{equation}{section}

\usepackage{filecontents}

\title[Generalized BGW tau-function as a hypergeometric BKP solution] {Generalized Br\'ezin--Gross--Witten tau-function as a hypergeometric solution of the BKP hierarchy}

\author{Alexander Alexandrov}

\address{IBS Center for Geometry and Physics,
	Pohang University of Science and Technology (POSTECH),
	77 Cheongam-ro, Nam-gu, Pohang, Gyeongbuk, 37673, Korea
}

\email{ {\tt alexandrovsash at gmail.com}}

\begin{document}

\begin{abstract}
In this paper, we prove that the generalized Br\'ezin--Gross--Witten tau-function is a hypergeometric solution of the BKP hierarchy with simple weight generating function.
We claim that it describes a spin version of the strictly monotone Hurwitz numbers. A family of the hypergeometric tau-functions of the BKP hierarchy, corresponding to the rational weight generating functions, is investigated. In particular, the cut-and-join operators are constructed, and the explicit description of the BKP Sato Grassmannian points is derived.
Representatives of this family can be associated with interesting families of spin Hurwitz numbers including a spin version of the 
monotone Hurwitz numbers.
\end{abstract}

\maketitle



\tableofcontents

\def\thefootnote{\arabic{footnote}}

\setcounter{equation}{0}

\section{Introduction}

Spin Hurwitz numbers were introduced by Eskin, Okounkov, and Pandharipande \cite{EOP}. Analogously to the ordinary Hurwitz numbers, they count the ramified coverings, but this time of the surfaces with spin structure. This is an interesting class of enumerative geometry invariants, closely related to the fundamental structures in representation theory and mathematical physics. In particular, the relation between spin Hurwitz numbers and Schur Q-functions was noticed by Gunningham \cite{Guni}. On the representation theory side the spin Hurwitz numbers are related to the characters of the Sergeev group.

Generating functions of the weighted double spin Hurwitz numbers were identified with the tau-functions of the 2-component BKP integrable hierarchy by Mironov, Morozov and Natanzon \cite{MMNQ}. Corresponding class  of the BKP tau-functions, so-called {\em hypergeometric} tau-functions was introduced and investigated by Orlov \cite{OBKP}. In this paper, we introduce and investigate an infinite parametric family of the  hypergeometric tau-functions of the 2-component BKP hierarchy. This family is associated with the rational weight generating functions.

The motivation for the investigation of this family is two-fold. On the one hand, the simplest representative of this family is given by the generalized Br\'ezin--Gross--Witten (BGW) tau-function. Relation between the generalized BGW tau-function and a hypergeometric solution of the BKP hierarchy was conjectured in  \cite{ABKP}, and its proof is one of the main results of this paper. On the other hand, we claim that this family contains several interesting examples of the generating functions of the weighted spin Hurwitz numbers, including the natural spin analogs of the monotone, strictly monotone, and Bousquet-M\'elou-Schaeffer Hurwitz numbers. In particular, the generalized BGW tau-function can be associated with the generating function of the strictly monotone spin Hurwitz numbers. Ordinary strictly monotone Hurwitz numbers describe Grothendieck's dessins d'enfants and hypermaps, and the spin case should also be related to the interesting invariants of enumerative geometry and combinatorics.

For the ordinary Hurwitz numbers, the rational weight generating functions can be obtained from the combination of the Schur functions with a particular choice of the variables. In the spin case, the correspondence is more complicated and will be considered elsewhere. Therefore, we do not discuss the detailed interpretation of the obtained generating functions in terms of spin Hurwitz numbers and relations between these numbers and intersection theory invariants. Instead in this paper, we focus on the integrable properties of this family. For this purpose we use the standard ingredients of the theory of integrable hierarchies including the free fermions, vertex operators, and the Sato Grassmannian. However, for the BKP and KP hierarchies, corresponding to the spin and ordinary Hurwitz numbers respectively, these ingredients are different, because of the difference between the underlying symmetry algebras, $\go(\infty)$ and $\gl(\infty)$.

In particular, we construct the basis for the corresponding points for the BKP Sato Grassmannian and derive associated Kac--Schwarz operators. Moreover, we construct the cut-and-join operators, which, for the polynomial weight generating functions, are the finite degree differential operators. The Kac--Schwarz description allows us to construct the quantum spectral curves and to derive the linear constraints for the tau-functions. We expect that these ingredients will be useful for further investigation of the spin Hurwitz numbers. In particular, they should be useful for the
construction of the Chekhov--Eynard--Orantin topological recursion. The description obtained in this paper is the spin analog of the description for the ordinary Hurwitz numbers, considered in \cite{ALS,AN, KZ,Zog}.

A few days after the current paper was posted on the arXiv, a new paper \cite{LC} appeared with an alternative proof of Theorem \ref{T3} via an action of the Virasoro operators on the Schur Q-functions.

\bigbreak

\noindent{\bf Notation.} In this paper  we denote the families of variables or parameters, finite or infinite, by bold symbols. In particular, $\bf t$ denotes the set of odd variables,
\be
{\bf t}=\{t_1,t_3,t_5,\dots\}.
\ee

A partition $\lambda$ is {\em strict}, if $\lambda_1>\lambda_2>\lambda_3>\dots>\lambda_{\ell(\lambda)}>\lambda_{\ell(\lambda)+1}=0$,  where $\ell(\lambda)$ is the {\em length} of the partition. We denote the set of strict partitions, including the empty one, by $\DP$.

\bigbreak

\noindent{\bf Acknowledgments.} This work was supported by IBS-R003-D1.  We thank the anonymous referee of the paper for the comments that led to a substantial
improvement of the paper.

\subsection{Organization of the paper} In Section \ref{SS2} we remind the reader the neutral fermion formalism and describe associated vertex operators. Section \ref{SS3} is devoted to the $w_{1+\infty}^B$ algebra and its central extension, $W_{1+\infty}^B$. In particular, we explicitly describe the convenient bases for these algebras. In Section \ref{S4} we outline the Sato Grassmannian description of the BKP hierarchy. In Section \ref{GBGW} we describe the diagonal group element, associated with the generalized BGW tau-function and prove that this tau-function is a hypergeometric solution of the BKP hierarchy. The generalization of this tau-function to the family with rational weight generating functions is investigated in Section \ref{S6}.

\section{Neutral fermions and boson-fermion correspondence}\label{SS2}

In this section we remind the reader the neutral fermion formalism and boson-fermion correspondence in the framework of the BKP hierarchy.
More details can be found in \cite{JMV,JMBKP,You,vdLASM,OBKP}.


\subsection{Neutral fermions}

Let $\phi_k$, $k\in {\mathbb Z}$ be the neutral free fermions satisfying the canonical anticommutation relations
\be\label{ac}
\left\{\phi_k,\phi_m\right\}=(-1)^k\delta_{k+m,0}.
\ee
Note that $\phi_0^2=1/2$. These relations define the Clifford algebra as the associative algebra.

For the vacuum $\rvac$ and the co-vacuum $\lvac$ vectors, satisfying
\be\label{vp1}
\phi_m \rvac =0,\,\,\,\,\, \lvac \phi_{-m}=0,\,\,\,\,\, m<0
\ee
the elements $\phi_{k_1}\phi_{k_2}\dots \phi_{k_m}\rvac$  with $k_1>k_2>\dots>k_m\geq0$ form a basis of a version of infinite wedge space (or neutral fermion Fock space) ${\mathcal F}_B$,
\be\label{bas1}
{\mathcal F}_B =\sppan\left\{ \phi_{k_1}\phi_{k_2}\dots \phi_{k_m}\rvac \ |\ k_1>k_2>\dots>k_m\geq 0 \right\}, 
\ee
and its dual
\be\label{bas2}
{\mathcal F}_B^* =\sppan\left\{\lvac  \phi_{k_m}\dots  \phi_{k_2}\phi_{k_1} \ |\  k_1<k_2<\dots<k_m\leq0  \right\}. 
\ee
The space ${\mathcal F}_B$ splits into two subspaces
\be
{\mathcal F}_B={\mathcal F}_B^0\oplus {\mathcal F}_B^1,
\ee
where ${\mathcal F}_B^0$ and ${\mathcal F}_B^1$ denote the subspaces with even and odd numbers of generators respectively. The same decomposition exists for ${\mathcal F}_B^*$.

There is a nondegenerate bilinear pairing ${\mathcal F}_B \times {\mathcal F}_B^* \rightarrow {\mathbb C}$, and the pairing of $\left<U\right | \in {\mathcal F}_B^*$ and $ \left|V\right>\in {\mathcal F}_B$ is denoted by $\left<U |V\right>$ with
\be
\left<0|0\right>=1.
\ee
 The {\em vacuum expectation value} of the Clifford algebra element $a$ is a pairing of $\lvac$ and $a\rvac$ which is denoted by  $\lvac a \rvac$. It is uniquely defined by the anticommutation relations (\ref{ac}), property (\ref{vp1}), and a relation 
 \be
 \lvac \phi_0 \rvac =0. 
 \ee
In particular, if $a$ is an odd element of the Clifford algebra, then $\lvac a \rvac =0$. It is easy to see that the bases in (\ref{bas1}) and  (\ref{bas2}) are orthogonal. 
Let us focus on the space ${\mathcal F}_B^0$. The basis can be labelled by strict partitions $\lambda \in \DP$ in the following way: 
\be\label{lambda}
\left|\lambda\right> = 
\begin{cases}
 \phi_{\lambda_1}\phi_{\lambda_2}\dots \phi_{\lambda_{\ell(\lambda)}}\rvac     & \mathrm{for}  \,\,\, \ell(\lambda)=0 \mod 2,\\
\sqrt{2} \phi_{\lambda_1}\phi_{\lambda_2}\dots \phi_{\lambda_{\ell(\lambda)}}\phi_0 \rvac  & \mathrm{for}   \,\,\, \ell(\lambda)=1 \mod 2.
\end{cases}
\ee
For the dual space ${\mathcal F}_B^{0*}$ we have a similar basis. From the anticommutation relations we have:

\begin{proposition} 
$\forall$ $\lambda, \mu \in \DP$
\be
\left<\lambda|\mu\right>=(-1)^{|\lambda|} \delta_{\lambda,\mu}.
\ee
\end{proposition}

It is easy to see that
\be
\lvac \phi_k\phi_{m} \rvac =\delta_{k+m,0}H[m], 
\ee
where
\be\label{Ha}
H[m]  =
\begin{cases}
\displaystyle{0}  & \mathrm{for} \quad m<0,\\
\displaystyle{\frac{1}{2}}  & \mathrm{for} \quad m=0,\\
\displaystyle{(-1)^m} \,\,\,\,\,\, & \mathrm{for} \quad m>0.
\end{cases}
\ee

Bilinear combinations of neutral fermions $\phi_k\phi_m$ satisfy the commutation relations of the Lie algebra $\go(\infty)$. Let $\left(E_{i,j}\right)_{k,l}=\delta_{i,k}\delta_{j,l}$ be the standard basis of the matrix units $\left.\left\{E_{{i,j}}\right| i,j \in {\mathbb Z}\right\}$. Then $\phi_k\phi_{-m}$ corresponds \cite{JMBKP} to 
\be
F_{k,m}=(-1)^m E_{k,m}-(-1)^k E_{-m,-k}
\ee
with the commutation relations
\be\label{comF}
\left[F_{a,b},F_{c,d}\right]=(-1)^b\delta_{b,c}F_{a,d}-(-1)^a\delta_{a+c,0}F_{-b,d}+(-1)^b\delta_{b+d,0}F_{c,-a}-(-1)^a\delta_{a,d}F_{c,b}.
\ee

For the bilinear combinations of neutral fermions we introduce the {\em normal ordering} by $\Normord{\phi_k\phi_{m}}=\phi_k\phi_{m}-\lvac \phi_k\phi_{m} \rvac$. It is skew-symmetric 
\be
\Normord{\phi_k\phi_{m}}=-\Normord{\phi_m\phi_{k}},
\ee
in particular, $\Normord{\phi_k\phi_{k}}=0$.
The normal ordered quadratic combinations of neutral fermions satisfy the commutation relations of a central extension of the algebra $\go(\infty)$
\begin{multline}
[\Normord{\phi_a\phi_b},\Normord{\phi_c\phi_d}]= (-1)^b \delta_{b+c,0}\Normord{\phi_a\phi_d} -(-1)^a\delta_{a+c,0}\Normord{\phi_b\phi_d}\\
+ (-1)^b\delta_{b+d,0}\Normord{\phi_c\phi_a} -(-)^a\delta_{a+d,0}\Normord{\phi_c\phi_b} + (\delta_{c,b}\delta_{a,d}-\delta_{a-c,0}\delta_{b-d,0})( (-1)^{a}H[b]- (-1)^{b}H[a]),
\end{multline}
where $H[a]$ is given by (\ref{Ha}).

Let us consider the generating function
\be
\phi(z)=\sum_{k\in {\mathbb Z}}\phi_k z^k. 
\ee
It satisfies the anticommutation relation 
\be\label{phid}
\left\{\phi(z),\phi(w)\right\}=\delta(z+w).
\ee
Here we introduce the delta-function
\be\label{delta}
\delta(z-w)=\sum_{k\in {\mathbb Z}}\left(\frac{z}{w}\right)^k.
\ee
It satisfies
\be
\delta(z-w)f(z)=\delta(z-w)f(w)
\ee
for any formal series $f(z)\in {\mathbb C}[\![z,z^{-1}]\!]$, and can be represented as
\be\label{delta1}
2\delta(z+w)=\frac{z-w}{\underline{z}+w}-\frac{z-w}{\underline{w}+z}.
\ee
Here $\frac{1}{\underline{z}-w}$ denotes the Laurent series expansion in the
region $|z|>|w|$,
\be
 \frac{1}{\underline{z}-w}:=\frac{1}{z}\sum_{k=0}^\infty \left(\frac{w}{z}\right)^k.
\ee

Quadratic combinations of these generating functions generate a Lie algebra with the following commutation relations
\begin{multline}\label{bilc}
\left[\phi_1(z_1)\phi(w_1),\phi(z_2)\phi(w_2)\right]=\delta(w_1+z_2)\phi(z_1)\phi(w_2)-\delta(z_1+z_2)\phi(w_1)\phi(w_2)\\
+\delta(w_1+w_2)\phi(z_2)\phi(z_1)-\delta(z_1+w_2)\phi(z_2)\phi(w_1).
\end{multline}
For the normal ordered operator we have
\be
\phi(z)\phi(w)=\Normord{\phi(z)\phi(w)}+\frac{1}{2} \frac{z-w}{\underline{z}+w}.
\ee


\subsection{Vertex operators}

For $k\in {\mathbb Z}_{\odd}$ we introduce the {\em bosonic} operators
\be
J_k=\frac{1}{2}\sum_{m\in {\mathbb Z}}(-1)^{m+1}\Normord{\phi_m\phi_{-m-k}}
\ee
satisfying a commutation relation of the Heisenberg algebra
\be\label{comJ}
\left[J_k,J_m\right]=\frac{k}{2}\delta_{k+m,0}.
\ee
From (\ref{vp1}) we have
\be
J_m \rvac =0,\,\,\,\,\, \lvac J_{-m}=0,\,\,\,\,\, m>0.
\ee

Let us consider the {\em vertex operator} for the BKP hierarchy introduced in \cite{JMBKP},
\be
\widehat{V}_B(z)=\exp\left(\sum_{k\in{\mathbb Z}_{\odd}^+}z^k t_k\right)\exp\left(-2\sum_{k\in{\mathbb Z}_{\odd}^+}\frac{1}{kz^k}\frac{\p}{\p t_k}\right).
\ee
These operators satisfy the anticommutation relations
\be\label{Vac}
\left\{\widehat{V}_B(z),\widehat{V}_B(w)\right\}=2 \delta(z+w)
\ee
similar to (\ref{phid}).

It is convenient to introduce the generating functions of the bosonic operators
\be
J_+({\bf t}) = \sum_{k\in {\mathbb Z}_{\odd}^+} t_k J_k,\,\,\,\,\,\,\, J_-({\bf s}) = \sum_{k\in {\mathbb Z}_{\odd}^+} s_k J_{-k}.
\ee
Then one has
\begin{equation}\label{Vtp}
\begin{split}
\widehat{V}_B(z)\lvac e^{J_+({\bf t})}&=2\lvac \phi_0  e^{J_+({\bf t})} \phi(z),\\
\widehat{V}_B(z)\lvac \phi_0 e^{J_+({\bf t})}&= \lvac e^{J_+({\bf t})} \phi(z).
\end{split}
\end{equation}

Let us consider a bilinear combination of the vertex operators $ \widehat{V}_B$,
\be
\widehat{V}_B(z,w)=\frac{1}{2}\widehat{V}_B(z) \widehat{V}_B(w).
\ee
Using the anticommutation relation (\ref{Vac}) it is easy to show that the vertex operators $\widehat{V}_B(z,w)$ satisfy the commutation relation, equivalent to the one described by (\ref{bilc}) for the bilinear combinations $\phi(z)\phi(w)$,
\begin{multline}\label{Vcommun}
\left[\widehat{V}_B(z_1,w_1),\widehat{V}_B(z_2,w_2)\right]= \delta(w_1+z_2)\widehat{V}_B(z_1,w_2)\\
-\delta(z_1+z_2)\widehat{V}_B(w_1,w_2)+\delta(w_1+w_2)\widehat{V}_B(z_2,z_1)-\delta(z_1+w_2)\widehat{V}_B(z_2,w_1).
\end{multline}
These vertex operators, dependent on two parameters, can be represented as
\be\label{Vr}
\widehat{V}_B(z,w)=\frac{1}{2}\frac{z-w}{\underline{z}+w}\exp\left(\sum_{k\in{\mathbb Z}_{\odd}^+}(z^k+w^k)t_k\right)\exp\left(-2\sum_{k\in{\mathbb Z}_{\odd}^+}\left(\frac{1}{kz^k}+\frac{1}{kw^k}\right)\frac{\p}{\p t_k}\right).
\ee

It is also convenient to consider a regularized version of the vertex operator, corresponding to $\Normord{\phi(z)\phi(-w)}$
\be\label{Ydd}
{\widehat Y}_B(z,w)=\widehat{V}_B(z,w)-\frac{1}{2}\frac{z-w}{\underline{z}+w},
\ee
or, equivalently
\be\label{Vreg}
{\widehat Y}_B(z,w)=\frac{1}{2}\left(\widehat{V}_B(z) \widehat{V}_B(w)-\frac{z-w}{\underline{z}+w}\right).
\ee
This expression is regular at $z=-w$, moreover, it is antisymmetric with respect to the permutation of $z$ and $w$,
\be\label{asim}
{\widehat Y}_B(z,w)=-{\widehat Y}_B(w,z).
\ee
From (\ref{Vtp}) it follows that
\be\label{ara}
{\widehat Y}_B(z,w) \lvac e^{J_+({\bf t}) }=\lvac e^{J_+({\bf t}) } \Normord{\phi(z)\phi(w)}.
\ee

\subsection{Boson-fermion correspondence}

For the neutral fermions the boson-fermion correspondence  \cite{You} describes an isomorphism
\be
\sigma_B^i:\,\,\,\,\, {\mathcal F}_B^i \simeq B^{(i)}={\mathbb C}[\![t_1,t_3,t_5,\dots ]\!]
\ee
for $i=0,1$. Here
\be
\sigma_B^i (\left| i\right> )=1,
\ee
where we introduce $\left|1\right>= \sqrt{2}\phi_0 \rvac$, and for both $i=0,1$ we have
\be
\sigma_B^i  J_{-k} (\sigma_B^i)^{-1}=\frac{k}{2}  t_k ,\,\,\,\,\,\,\, \sigma_B^i  J_k (\sigma_B^i)^{-1} = \frac{\p}{\p t_k}
\ee
for $k\in{\mathbb Z}_{\odd}^+$. The boson-fermion correspondence is given by
\be
\sigma_B^i(\left| a \right>)=\begin{cases}
\displaystyle{\left<1\right| e^{J_+({\bf t}) } \left| a \right>} \,\,\,\,\,\,\,\,\,\,\,\,\,\,\,\,\, \mathrm{for} \quad  \left| a \right> \in  {\mathcal F}_B^1,\\[4pt]
\displaystyle{\lvac e^{J_+({\bf t}) } \left| a \right> }\,\,\,\,\,\,\,\,\,\,\,\,\,\,\,\,\, \mathrm{for} \quad \left| a \right> \in  {\mathcal F}_B^0,
\end{cases}
\ee
where $\left<1\right|=\sqrt{2}\lvac \phi_0$. The boson-fermion correspondence between two different representations of the central extension of the $\go(\infty)$ algebra is given by
\be
\sigma_B^i \Normord{\phi(z)\phi(w)} (\sigma_B^i )^{-1}={\widehat Y}_B(z,w).
\ee
Below we will work only with ${\mathcal F}_B^0$ component of the fermionic Fock space and its bosonic counterpart. For them we denote the boson-fermion correspondence by $\sigma$.

\section{\texorpdfstring{$W_{1+\infty}^B$}--algebra}\label{SS3}

Content of this section is closely related to the results of van de Leur \cite{vdLASM}.

\subsection{\texorpdfstring{$w_{1+\infty}^B$}--algebra and its central extension}
Let us introduce the algebra $W^B_{1+\infty}$. Consider $w_{1+\infty}$, the algebra of diffeomorphisms  on the circle, 
\be
w_{1+\infty}=\sppan\{z^k\p_z^m \ | \ k\in {\mathbb Z}, m\in {\mathbb Z}_{\geq 0}\}.
\ee
Let $\iota$ be the anti-involution of $w_{1+\infty}$
\be
\iota (z)=-z, \,\,\,\,\, \iota (z\p_z)=-z \p_z.
\ee
We define
\be
w_{1+\infty}^B=\{ {\mathtt a} \in w_{1+\infty} \ |\ \iota({\mathtt a})=-{\mathtt a} \}.
\ee

For any $ {\mathtt a}  \in w_{1+\infty}$ we denote by $\bar{{\mathtt a}}$ the operator, obtained by the sign inversion of $z$,
\be
\overline{ {\mathtt a} }= {\mathtt a} \big|_{z\mapsto -z, \p_z\mapsto -\p_z}.
\ee
Then it is easy to see that 
\be\label{iott}
\iota ( {\mathtt a} )= \overline{ {\mathtt a} ^*},
\ee
where ${\mathtt a}^*\in w_{1+\infty}$ is the adjoint operator for which an identity
$$ 
\res_z \left(z^{-1} f(z)\,  {\mathtt a} \, g(z)  \right)
=\res_z \left(z^{-1} g(z)\,  {\mathtt a} ^*\, f(z) \right)
$$
holds for any commuting $f(z)$ and $g(z)$. Here 
\be
\res_z \sum_{k \in {\mathbb Z}} a_k z^k:= a_{-1}.
\ee
In particular, 
\be
(z^k\p_z^m)^*=z(-\p_z)^m z^{k-1}.
\ee
\begin{definition}\label{def1}
For any operator $ {\mathtt a} \in w_{1+\infty}$ we introduce a differential operator, acting in the bosonic Fock space ${\mathbb C}[\![t_1,t_3,t_5\dots ]\!]$:
\be
\widehat{W}_ {\mathtt a} ^B=\frac{1}{2}\res_w w^{-1} \left. {\mathtt a} _w \cdot {\widehat Y}_B(z,w)\right|_{z=-w}. 
\ee
\end{definition}
Here by $ {\mathtt a} _w$ we denote the operator $ {\mathtt a} $ acting on the space of functions of the variable $w$.

\begin{lemma}
For $ {\mathtt a} \in w_{1+\infty}$
\be
\left[\widehat{V}_B(z), \widehat{W}_{\mathtt a}^B\right]=\frac{1}{2}( {\mathtt a} _z-\iota( {\mathtt a} _z)) \cdot \widehat{V}_B(z).
\ee
\end{lemma}
\begin{proof}
From the anticommutation relation (\ref{Vac}) we have
\be
\left[\widehat{V}_B(z),\widehat{Y}_B(v,w)\right]=\delta(z+v)\widehat{V}_B(w)-\delta(z+w)\widehat{V}_B(v).
\ee
Hence,
\begin{equation}
\begin{split}
\left[ \widehat{V}_B(z),\widehat{W}_ {\mathtt a} ^B\right]&= \frac{1}{2}\res_w w^{-1} \left. {\mathtt a} _w \left(\delta(z+v)\widehat{V}_B(w)-\delta(z+w)\widehat{V}_B(v)\right) \right|_{v=-w} \\
&= \frac{1}{2}\res_w w^{-1} \left( \delta(z-w)  {\mathtt a} _w \widehat{V}_B(w) -  \delta(z+w)  {\mathtt a} _w^* \widehat{V}_B(-w)\right),
\end{split}
\end{equation}
and the statement of the lemma follows from (\ref{iott}).
\end{proof}

\begin{corollary}\label{cor1}
For $ {\mathtt a}  \in w_{1+\infty}^B$
\be
\left[\widehat{V}_B(z),\widehat{W}_ {\mathtt a} ^B\right]= {\mathtt a}_z \cdot \widehat{V}_B(z).
\ee
\end{corollary}

From the commutation relation (\ref{Vcommun}) and definition of the regularized vertex operator (\ref{Ydd}) we have
\begin{multline}\label{commm}
\left[\widehat{Y}_B(z_1,w_1),\widehat{Y}_B(z_2,w_2)\right]= \delta(w_1+z_2)\widehat{Y}_B(z_1,w_2)\\
-\delta(z_1+z_2)\widehat{Y}_B(w_1,w_2)+\delta(w_1+w_2)\widehat{Y}_B(z_2,z_1)-\delta(z_1+w_2)\widehat{Y}_B(z_2,w_1)\\
+\frac{1}{4}\left(\frac{w_1-z_2}{\underline{w_1}+z_2}\frac{z_1-w_2}{\underline{z_1}+w_2}- \frac{z_1-z_2}{\underline{z_1}+z_2}\frac{w_1-w_2}{\underline{w_1}+w_2}+\frac{w_2-w_1}{\underline{w_2}+w_1}\frac{z_2-z_1}{\underline{z_2}+z_1}-\frac{w_2-z_1}{\underline{w_2}+z_1}\frac{z_2-w_1}{\underline{z_2}+w_1}\right).
\end{multline}
Here we use (\ref{delta1}).

Let
\be
\mu({\mathtt a},{\mathtt b})=\frac{1}{8}\res_{w_1} \res_{w_2} w_1^{-1} w_2^{-1} \frac{w_1-w_2}{\underline{w_1}+w_2}\left({\mathtt a}_{w_2}{\mathtt b}_{w_1}-{\mathtt a}_{w_1}{\mathtt b}_{w_2}\right) \frac{w_1-w_2}{\underline{w_1}+w_2}.
\ee
Then using (\ref{commm}), from a direct computation we see that
operators $\widehat{W}_ {\mathtt a} ^B$ describe algebra $W_{1+\infty}^B$, a central extension of the algebra $w_{1+\infty}^B$:
\begin{lemma}\label{Wcomm}
For $ {\mathtt a} , {\mathtt b} \in w_{1+\infty}^B$
\be
\left[\widehat{W}_ {\mathtt a} ^B,\widehat{W}_ {\mathtt b} ^B\right]=\widehat{W}^B_{[ {\mathtt a} , {\mathtt b} ]}+\mu( {\mathtt a} , {\mathtt b}).
\ee
\end{lemma}

From (\ref{ara}) we have
\be\label{Wferm1}
\widehat{W}_{\mathtt a}^B\cdot  \lvac e^{J_+({\bf t}) }=\lvac e^{J_+({\bf t}) } W_{\mathtt a}^B,
\ee
where
\be\label{Wferm}
W_{\mathtt a}^B=\frac{1}{2} \res_w w^{-1} \Normord{\phi(-w){\mathtt a}_w \phi(w)}
\ee
is the bilinear fermionic operator acting on the fermionic Fock space.

\subsection{Bases in \texorpdfstring{$w_{1+\infty}^B$}n-algebra and its central extension}

Let us consider the operators
\be
{\mathtt w}_{k,m}:=-z^{k+m}(\p_z)^m+(-1)^{k+m}z (\p_z)^m z^{k+m-1}
\ee
for $k \in {\mathbb Z}$, $m\in {\mathbb Z}_{\geq 0}$. It is easy to see that ${\mathtt w}_{k,m}\in w_{1+\infty}^B$. Operators with 
$k+m \in {\mathbb Z}_{\odd}$ constitute a basis in $w^B_{1+\infty}$, for example
\begin{equation}\label{wmin}
 \begin{aligned}
{\mathtt w}_{k,0}&=-2 z^k,                                                                                 &k &\in{\mathbb Z}_{\odd}, \\
{\mathtt w}_{k,1}&=-2z^{k+1}\p_z-k z^k,                                                           &k  & \in{\mathbb Z}_{\even},\\
{\mathtt w}_{k,2}&=-2z^{k+2}\p_z^2-2(k+1)z^{k+1}\p_z-(k+1)kz^k, \,\,\,\,\,\,\,\, &k &\in{\mathbb Z}_{\even}.\\
\end{aligned}
\end{equation}

Associated basis of $W_{1+\infty}^B$ consists of operators
\be
\widehat{W}_{k,m}^B:=\widehat{W}_{{\mathtt w}_{k,m}}^B.
\ee
for $k+m \in {\mathbb Z}_{\odd}$, and a central element.

Let
\be
\widehat{J}_B(z)=\sum_{k\in{\mathbb Z}_{\odd}^+}\left( k t_k z^{k-1}+2 z^{-k-1} \frac{\p}{\p t_k}\right) .
\ee
We introduce the {\em bosonic normal ordering} $\normord \dots \normord $ which puts all $t_k$ to the left of all $\frac{\p}{\p t_k}$.
Consider the operators 
\be
U_k(z)=\normord P_k(\widehat{J}_B(z)) \normord. 
\ee
Here $P_m$ are the Fa\`a di Bruno differential polynomials:
\be
P_k(\varphi') = e^{-\varphi(z)}\p_z^k e^{\varphi(z)} =(\p_z +\varphi')^k,
\ee
where $\varphi'=\p_z \varphi(z)$. They have a simple expression in terms of elementary Schur functions $p_k$,
\be
P_k( \varphi')=k!p_k( \varphi^{(j)}(z)/j!),
\ee
because
\be
\sum_{k=0}^\infty \frac{s^k}{k!} P_k(\p_z \varphi)= e^{-\varphi(z)} e^{\varphi(z+s)}=e^{\sum_{j=1}^\infty \frac{s^j}{j!} \varphi^{(j)}(z)}=\sum_{k=0}^\infty s^k p_k( \varphi^{(j)}(z)/j!).
\ee
Let
\be
\widehat{\varphi}_B(z)= \sum_{k\in{\mathbb Z}_{\odd}^+}\left( t_k z^k -2 \frac{1}{k z^k} \frac{\p}{\p t_k}\right),
\ee
then $\widehat{J}_B(z)=\p_z \widehat{\varphi}_B(z)$ and the vertex operator (\ref{Vreg}) can be expressed as
\be
{\widehat Y}_B(z,w)=\frac{1}{2}\frac{z-w}{\underline{z}+w} \normord e^{\widehat{\varphi}_B(z)+\widehat{\varphi}_B(w)}-1 \normord.
\ee

The following lemma describes operators $\widehat{W}_{k,m}^B$ in terms of the  Fa\`a di Bruno polynomials 
\begin{lemma} 
For $k \in {\mathbb Z}$, $m\in {\mathbb Z}_{\geq 0}$
\be
\widehat{W}_{k,m}^B=\frac{1}{2}\res_w w^{k+m-1} \left(\frac{2w}{m+1} U_{m+1}(w)+U_{m}(w)\right),
\ee
where we put $U_0(w)=0$.
\end{lemma}
\begin{proof}
By definition
\begin{equation}
\begin{split}
\widehat{W}_{k,m}^B&=\frac{1}{2}\res_w w^{-1} \left.\left((-1)^{k+m}w \p_w^m w^{k+m-1}-w^{k+m} \p_w^m\right) {\widehat Y}_B(z,w)\right|_{z=-w} \\
&=\frac{1}{2}\res_w w^{-1} \left.\left(z^{k+m} \p_z^m-w^{k+m} \p_w^m\right) {\widehat Y}_B(z,w)\right|_{z=-w}.
\end{split}
\end{equation}
It is easy to see that for any series $f(z,w)\in {\mathbb C}[\![z,w]\!]$ such that $f$ is antisymmetric, $f(z,w)=-f(w,z)$, one has
\be
\res_w w^{-1} \left.z^k \p_z^m f(z,w)\right|_{z=-w}=-\res_w w^{-1} \left.w^k \p_w ^m f(z,w)\right|_{z=-w}.
\ee
Since ${\widehat Y}_B(z,w)$ is antisymmetric, (\ref{asim}), we conclude that
\be
\widehat{W}_{k,m}^B=-\res_w  \left. w^{k+m-1}(\p_w)^m {\widehat Y}_B(z,w)\right|_{z=-w}.
\ee 
Then
\begin{equation}
\begin{split}
 \left. w^k \p_w^m {\widehat Y}_B(z,w)\right|_{z=-w}&= \left.  w^k \p_\Delta^m  {\widehat Y}_B(-w,w+\Delta)\right|_{\Delta=0}\\
 &= \left. -\frac{1}{2} w^k \p_\Delta^m (2w+\Delta) \sum_{j=1}^\infty \frac{\Delta^{j-1}}{j!} U_j(w) \right|_{\Delta=0}\\
 &=-\frac{1}{2}w^k \left(\frac{2w}{m+1} U_{m+1}(w)+U_{m}(w)\right),
\end{split}
\end{equation}
where for $m=0$ we assume $U_0(w)=0$.
This concludes the proof.
\end{proof}

The {\em Heisenberg-Virasoro subalgebra} of the algebra $W_{1+\infty}^B$ is generated by the operators 
\be\label{JB}
\widehat{J}_k^B =
\begin{cases}
 \begin{aligned}
 2 \frac{\p}{\p t_k} \,\,\,\,\,\,\,\,\,\,\,\,\,\,\,\,\,\,\, &\mathrm{for} \quad k \in {\mathbb Z}_{\odd}^+,\\[5pt]
-kt_{-k}\,\,\,\,\,\,\,\,\,\,\,\,\,\,\,&\mathrm{for} \quad -k\in {\mathbb Z}_{\odd}^+,\\[5pt]
0 \,\,\,\,\,\,\,\,\,\,\,\,\,\,\,\,\,\,\,\,\,\,\,\,\,\,\,\,\,&\mathrm{overwise} 
 \end{aligned}
\end{cases}
\ee
and the Virasoro operators
\be
\widehat{L}_k^B=
\begin{cases}
 \begin{aligned}
&\frac{1}{2}\sum_{i+j=k} \normord\widehat{J}_i^B\widehat{J}_j^B\normord,  \quad  \quad &\mathrm{for} &\quad  k \in  {\mathbb Z}_{\even},\\[6pt]
&0,  \quad  \quad  \quad  \quad  \quad  \quad \quad \quad  \quad \,\,\,\,\,& \mathrm{for}&  \quad  k \in  {\mathbb Z}_{\odd}.
 \end{aligned}
\end{cases}
\ee
Here the bosonic normal ordering puts all $\widehat{J}_m^B$ with positive $m$ to the right of all $\widehat{J}_m^B$ with
negative $m$.
Let us also consider the $W^{(3)}$-algebra, which includes the generators
\be\label{MB}
\widehat{M}_k^B=
\begin{cases}
 \begin{aligned}
&\frac{1}{3}\sum_{i+j+l=k}  \normord\widehat{J}_i^B\widehat{J}_j^B\widehat{J}_l^B \normord,  \quad  \quad &\mathrm{for}& \quad  k \in  {\mathbb Z}_{\odd},\\[6pt]
&0,  \quad  \quad  \quad  \quad  \quad  \quad \quad \quad  \quad \quad \quad \quad &\mathrm{for} & \quad  k \in  {\mathbb Z}_{\even}.
 \end{aligned}
\end{cases}
\ee
These operators satisfy the following commutation relations
\begin{equation}
 \begin{aligned}
\left[\widehat{J}_k^B,\widehat{J}_m^B\right]&=2\delta_{k+m,0} k,                                                                                         & k&,m \in{\mathbb Z}_{\odd}, \\
\left[\widehat{J}_k^B,\widehat{L}_m^B\right]&=2 k\widehat{J}_{k+m}^B,                                                                              & k&  \in{\mathbb Z}_{\odd}, m  \in{\mathbb Z}_{\even},\\
\left[\widehat{L}_k^B,\widehat{L}_m^B\right]&=2(k-m)\widehat{L}_{k+m}^B+\frac{1}{3}k(k^2-1)\delta_{k,-m},                     & k&,m \in{\mathbb Z}_{\even},\\
\left[\widehat{J}_k^B,\widehat{M}_m^B\right]&=4k\,\widehat{L}_{k+m}^B.                                                                            &  k&,m \in{\mathbb Z}_{\odd},\\
\left[\widehat{L}_k^B,\widehat{M}_m^B\right]&=2(2k-m)\widehat{M}_{k+m}+\frac{2}{3}k(k^2-1)\widehat{J}_{k+m}, \,\,\,\,\, & k&  \in{\mathbb Z}_{\even}, m  \in{\mathbb Z}_{\odd},\\
\end{aligned}
\end{equation}
and, except for the case of two operators $\widehat{M}_m^B$, these operators commute otherwise.
A commutator of two $\widehat{M}_k^B$'s contains the terms of fourth power of the bosonic operators $\widehat{J}_m^B$, so it can not be represented as a linear combination of $\widehat{J}_k^B$, $\widehat{L}_k^B$, and $\widehat{M}_k^B$.

For operators $\widehat{W}_{k,m}^B$ with $m\leq 2$ we have
\begin{equation}
\begin{split}\label{Wop}
\widehat{W}_{k,0}^B&=  \widehat{J}_k^B,\\
\widehat{W}_{k,1}^B&=  \widehat{L}_k^B-\frac{k}{2} \widehat{J}_k^B,\\
\widehat{W}_{k,2}^B&= \widehat{M}_k^B-(k+1) \widehat{L}_k^B+\frac{(k+1)(2k+1)}{6}  \widehat{J}_k^B.
\end{split}
\end{equation}

\section{Orthogonal Sato Grassmannian}\label{S4}

In this section we briefly summarize some properties of the Sato Grassmannian. It describes the space of solutions of the KP hierarchy \cite{Sato} and its BKP version. We describe the action of the algebra $w_{1+\infty}^B$ on the BKP Sato Grassmannian and clarify the difference between the points of the KP and BKP Sato Grassmannians for the solutions of the KdV hierarchy.

 \subsection{Symmetries of the BKP Sato Grassmannian}
 Let us consider the space $H=H_+\oplus H_-$, where the subspaces 
\be
H_-=z^{-1}{\mathbb C}[\![z^{-1}]\!]
\ee and 
\be
H_+={\mathbb C}[z]
\ee 
are generated by negative and nonnegative powers of $z$ respectively. Then the Sato Grassmannian $\rm{Gr}$ consists of all closed linear spaces $\mathcal{W}\in H $, which are compatible with $H_+$. Namely, an orthogonal projection $\pi_+ : \mathcal{W} \to H_+ $ should be a Fredholm operator, i.e. both the kernel ${\rm ker}\, \pi_+ \in \mathcal{W}$ and the 
cokernel ${\rm coker}\, \pi_+ \in H_+$ should be finite-dimensional vector spaces. 
The Grassmannian $\Gr$ consists of components $\Gr^{(k)}$, parametrized by an index of the operator $\pi_+$. Below we use only the component $\Gr^{(0)}$; other components have an equivalent description. The big cell $\Gr^{(0)}_+$ of $\Gr^{(0)}$  is defined by the constraint ${\rm ker}\, \pi_+ = {\rm coker}\, \pi_+=0$. We call $\Gr^{(0)}_+$ the {\em Sato Grassmannian} for simplicity.

A point of the Sato Grassmannian ${\mathcal W}\in \rm{Gr}^{(0)}_+$ can be described by an  {\emph {admissible basis}} 
\be
{\mathcal W}=\sppan_{\mathbb C}\{\Phi_1^{\mathcal W},\Phi_2^{\mathcal W},\Phi_3^{\mathcal W},\dots\}.
\ee
Let 
\be
\Psi={\mathcal W}\cap 1+z^{-1}{\mathbb C}[\![z^{-1}]\!].
\ee
be the {\em wave function}.

The BKP hierarchy can be represented in terms of tau-function $\tau({\bf t})$ by the Hirota bilinear identity
\be\label{HBE}
\res_z z^{-1} e^{\sum_{k\in{\mathbb Z}_{\odd}^+}z^k (t_k-t'_k)} 
\tau ({\bf t}-2[z^{-1}])\tau ({\bf t'}+2[z^{-1}]) =\tau({\bf t})\tau({\bf t'}).
\ee
Here ${\bf t}$ and ${\bf t'}$ are two independent sets of variables.
For
$f(z) \in H$ and $g(z) \in H$ put
\be
(f(z),g(z))_B := \res_z z^{-1} f (z)g (-z).
\ee
Note that $(\ ,\ )_B$ is a non-degenerate symmetric bilinear pairing on
$H$. For ${\mathcal W}\in \Gr^{(0)}$ we denote by ${\mathcal W}^{\bot_B}$ the orthogonal compliment of ${\mathcal W}$ with respect to $(\ ,\ )_B$.
The BKP (or orthogonal) Sato Grassmannian $\Gr^{(0)}_B$ is the subspace of $\Gr^{(0)}$ such that ${\mathcal W}^{\bot_B} \subset {\mathcal W}$ and the quotient $ {\mathcal W}/{\mathcal W}^{\bot_B}$ is one dimensional and is generated by $\Psi$. Similarly to the KP case, one can consider the big cell of the BKP Sato Grassmannian, $\Gr^{(0)}_{B+} = \Gr^{(0)}_{+} \cap \Gr^{(0)}_{B}$.
There exists a bijection between the points of the BKP Sato Grassmannian ${\mathcal W}\in\Gr^{(0)}_{B+}$ and the BKP tau-functions with $\tau({\bf 0})=1$.  
Below we focus on $\Gr^{(0)}_{B+}$, and put $\tau({\bf 0})=1$ for simplicity.

Tau-functions of the BKP hierarchy are given by the vacuum expectation values of neutral fermions
\be\label{taug}
\tau_{G}({\bf t})=\lvac e^{J_+({\bf t})}  G \rvac.
\ee
They are labeled by the group element of the centrally extended $\go(\infty)$ algebra,
\be
G=\exp\left(\sum_{k,m\in {\mathbb Z}} a_{km}\Normord{\phi_k \phi_m}\right).
\ee
The wave function is equal to the principal specialization of the BKP tau-function
\be
\Psi(z)=\tau(-2[z^{-1}])=\tau(2[-z^{-1}]),
\ee
because the BKP tau-functions depend only on odd times. The {\em Baker--Akhiezer function}
\be
\Psi^B(z,{\bf t})= e^{\sum_{k\in{\mathbb Z}_{\odd}^+}t_k z^k} \frac{\tau({\bf t}-2[z^{-1}])}{\tau({\bf t})}
\ee
generates the corresponding point of the BKP Sato Grassmannian through the ${\bf t}$ series expansion 
\be
\Psi(z,{\bf t})\in{\mathcal W}.
\ee
Let us label the Baker--Akhiezer function, associated with the tau-function (\ref{taug}), by $G$. Then, from (\ref{Vtp}) we have
\be
\Psi_G^B(z,{\bf t})=\frac{2 \lvac \phi_0 e^{J_+({\bf t})} \phi(z) G \rvac}{\tau_G({\bf t})}.
\ee

For some ${\mathtt a}\in w_{1+\infty}^B$ let us consider 
\be
\tilde{G}=e^{W_{\mathtt a}^B} G,
\ee
where $W_{\mathtt a}^B$ is given by (\ref{Wferm}). We assume that ${\mathtt a}$ and $G$ are such that both $\tau_G$ and $\tau_{\tilde{G}}$ are well defined. 
Then from (\ref{Vtp}) for the associated wave function we have
\begin{equation}
\begin{split}
\Psi_{\tilde{G}}^B(z,{\bf t})&=\frac{\widehat{V}_B(z) \lvac e^{J_+({\bf t})} e^{{W}_{\mathtt a}^B}G \rvac }{\tau_{\tilde{G}}({\bf t})}\\
&=\frac{\widehat{V}_B(z) e^{\widehat{W}_{\mathtt a}^B }\lvac e^{J_+({\bf t})} G \rvac }{\tau_{\tilde{G}}({\bf t})}.
\end{split}
\end{equation}
From Corollary \ref{cor1} we have the operator identity
\be
\widehat{V}_B(z) e^{\widehat{W}_{\mathtt a}^B}= e^{\widehat{W}_{\mathtt a}^B} \left(e^{-\widehat{W}_{\mathtt a}^B} \widehat{V}_B(z) e^{\widehat{W}_{\mathtt a}^B}\right)= e^{\widehat{W}_{\mathtt a}^B} \left( e^{{\mathtt a}_z} \widehat{V}_B(z) \right),
\ee
hence
\begin{equation}
\begin{split}
\Psi^B_{\tilde{G}}(z,{\bf t})&= \frac{ e^{\widehat{W}_{\mathtt a}^B} e^{{\mathtt a}_z}\widehat{V}_B(z)   \lvac e^{J_+({\bf t})} G \rvac }{\tau_{\tilde{G}}({\bf t})}\\
&= \frac{ e^{\widehat{W}_{\mathtt a}^B}   \tau_G({\bf t}) e^{a_z}   \Psi^B_G(z,{\bf t}) }{\tau_{\tilde{G}}({\bf t})}.
\end{split}
\end{equation}
Inverting the operator $e^{\widehat{W}_{\mathtt a}^B} $ we have
\be
e^{{\mathtt a}_z}  \Psi^B_G(z,{\bf t})=  \frac{ e^{-\widehat{W}_{\mathtt a}^B}   \tau_{\tilde{G}}({\bf t})  \Psi^B_{\tilde{G}}(z,{\bf t}) }{\tau_{G}({\bf t})}.
\ee
Since $\Psi^B_{\tilde G}(z,{\bf t}) \in {\mathcal W}_{\tilde G}$, we also have
\be
  \frac{ e^{-\widehat{W}_{\mathtt a}^B}   \tau_{\tilde{G}}({\bf t})  \Psi^B_{\tilde{G}}(z,{\bf t}) }{\tau_{G}({\bf t})}
 \in {\mathcal W}_{\tilde{G}}.
 \ee
Therefore,
\be
e^{{\mathtt a}}\cdot {\mathcal W}_{G} \subset {\mathcal W}_{\tilde{G}},
\ee
and inverse is also true
\be
e^{-{\mathtt a}}\cdot {\mathcal W}_{\tilde{G}} \subset {\mathcal W}_{G}.
\ee
Then the following lemma justifies Definition \ref{def1}:
\begin{lemma}\label{L41}
For ${\mathtt a}\in w_{1+\infty}^B$ and ${\mathcal W}_{G} \in \Gr^{(0)}_{B}$
\be
e^{{\mathtt a}}\cdot {\mathcal W}_{G} = {\mathcal W}_{e^{{W}_{\mathtt a}^B}G}.
\ee
\end{lemma}


\subsection{KdV vs BKP}\label{kdvb}
 In \cite{AKdV} we prove that after a simple redefinition of times any tau-function of the KdV hierarchy solves the BKP hierarchy:
\begin{theorem}[\cite{AKdV}]\label{MT}
For any KdV tau-function  
\be\label{ms}
\tau_{\BKP}({\bf t})=\tau_{\KdV}({\bf t}/2)
\ee 
is a tau-function of the BKP hierarchy.
\end{theorem}

The KdV hierarchy is a 2-reduction of the KP hierarchy. 
Let us consider a tau-function of KdV hierarchy $\tau_{\KdV}$ and the point ${\mathcal W}^{\KdV}\in \Gr^{(0)}$ associated to it as a tau-function of the KP hierarchy. If we consider the same tau-function as a tau-function of the BKP hierarchy, in general it defines another point of the same Sato Grassmannian ${\mathcal W}^{\BKP}\in \Gr^{(0)}_B \subset  \Gr^{(0)}$. The reason is that the
Baker--Akhiezer functions of two hierarchies have different form.  

Indeed, as a tau-function of the KP hierarchy, $\tau_{\KdV}$ defines the 
Baker--Akhiezer function by
\be
\Psi^{\KdV}(z,{\bf t})=e^{\sum_{k\in {\mathbb Z}_{\odd}^+} t_k z^k} \frac{\tau_{\KdV}({\bf t}-[z^{-1}])}{\tau_{\KdV}(\bf t)},
\ee
where we put $t_{2k}=0$ for $k>0$.
If we represent it in terms of the corresponding BKP tau-function using Theorem \ref{MT}, we get
\be
\Psi^{\KdV}(z,{\bf t}/2)=e^{\frac{1}{2}\sum_{k\in {\mathbb Z}_{\odd}^+} t_k z^k}  \frac{\tau_{\BKP}({\bf t}-2[z^{-1}])}{\tau_{\BKP}(\bf t)}.
\ee
The Baker--Akhiezer function of the same tau-function, considered in the context of the BKP hierarchy, is
\be
\Psi^{\BKP}(z,{\bf t})=e^{\sum_{k\in {\mathbb Z}_{\odd}^+} t_k z^k}  \frac{\tau_{\BKP}({\bf t}-2[z^{-1}])}{\tau_{\BKP}(\bf t)}.
\ee
If ${\mathcal W}^{\BKP}\in \Gr^{(0)}_{B+}$ then at ${\bf t}={\bf 0}$ these two Baker--Akhiezer functions coincide
\be\label{firstb}
\Phi_1^{\KdV}(z)=\Phi_1^{\BKP}(z)=\tau_{\BKP}(-2[z^{-1}])=:\Psi(z).
\ee
However, already the next basis vectors are different:
\begin{equation}
\begin{split}
\Phi_2^{\KdV}(z)&=\left.2 \frac{\p}{\p t_1}\Psi^{\KdV}(z,{\bf t}/2)\right|_{{\bf t}={\bf 0}}\\
&=z\Psi(z)+\left.2 \frac{\p}{\p t_1}\tau_{\BKP}({\bf t}-2[z^{-1}])\right|_{{\bf t}={\bf 0}}-2\Psi(z)\left. \frac{\p}{\p t_1}\tau_{\BKP}({\bf t})\right|_{{\bf t}={\bf 0}},
\end{split}
\end{equation}
and
\begin{equation}
\begin{split}
\Phi_2^{\BKP}(z)&=\left. \frac{\p}{\p t_1}\Psi^{\BKP}(z,{\bf t}/2)\right|_{{\bf t}={\bf 0}}\\
&=z\Psi(z)+\left. \frac{\p}{\p t_1}\tau_{\BKP}({\bf t}-2[z^{-1}])\right|_{{\bf t}={\bf 0}}-\Psi(z)\left. \frac{\p}{\p t_1}\tau_{\BKP}({\bf t})\right|_{{\bf t}={\bf 0}}.
\end{split}
\end{equation}
Moreover, it is clear that this difference, in general, cannot be compensated by $\Psi$. Hence, in general, 
\be
{\mathcal W}^{\KdV}\neq {\mathcal W}^{\BKP}.
\ee
Therefore, the Kac--Schwarz algebras of a KdV tau-function, considered in the frameworks of the KP and BKP hierarchies, are different.

\section{Generalized Br\'ezin--Gross--Witten model}\label{GBGW}

\subsection{Intersection theory}
Denote by $\overline {\mathcal M}_{g,n}$ the Deligne--Mumford compactification of the moduli space of all compact Riemann surfaces of genus~$g$ with~$n$ distinct marked points. It is a non-singular complex orbifold of dimension~$3g-3+n$. It is empty unless the stability condition
\begin{gather}\label{stability}
2g-2+n>0
\end{gather}
is satisfied. 

For each marking index~$i$ consider the cotangent line bundle ${\mathbb{L}}_i \rightarrow \overline{\mathcal{M}}_{g,n}$, whose fiber over a point $[\Sigma,z_1,\ldots,z_n]\in \overline{\mathcal{M}}_{g,n}$ is the complex cotangent space $T_{z_i}^*\Sigma$ of $\Sigma$ at $z_i$. Let $\psi_i\in H^2(\overline{\mathcal{M}}_{g,n},\mathbb{Q})$ denote the first Chern class of ${\mathbb{L}}_i$. We consider the intersection numbers
\begin{gather}\label{eq:products}
\<\tau_{a_1} \tau_{a_2} \cdots \tau_{a_n}\>_g:=\int_{\overline{\mathcal{M}}_{g,n}} \psi_1^{a_1} \psi_2^{a_2} \cdots \psi_n^{a_n}.
\end{gather}
The integral on the right-hand side of~\eqref{eq:products} vanishes unless the stability condition~\eqref{stability} is satisfied, all  $a_i$ are non-negative integers, and the dimension constraint 
\be\label{d1}
3g-3+n=\sum_{i=1}^n a_i
\ee 
holds true. Let $T_i$, $i\geq 0$, be formal variables and let
\be
\tau_{KW}:=\exp\left(\sum_{g=0}^\infty \sum_{n=0}^\infty \hbar^{2g-2+n}F_{g,n}\right),
\ee
where
\be
F_{g,n}:=\sum_{a_1,\ldots,a_n\ge 0}\<\tau_{a_1}\tau_{a_2}\cdots\tau_{a_n}\>_g\frac{\prod T_{a_i}}{n!}.
\ee
Witten's conjecture \cite{Wit91}, proved by Kontsevich \cite{Kon92}, states that the partition function~$\tau_{KW}$ becomes a tau-function of the KdV hierarchy after the change of variables~$T_n=(2n+1)!!t_{2n+1}$. This is the Kontsevich-Witten tau-function.

In this paper we will focus on a different version of intersection theory on the moduli spaces, which is also governed by the KdV hierarchy, and was recently introduced by  Norbury \cite{Norb}.
Namely, he introduced $\Theta$-classes, $\Theta_{g,n}\in H^{4g-4+2n}(\overline{\mathcal{M}}_{g,n})$, and described their intersections with the $\psi$-classes
\be
\<\tau_{a_1} \tau_{a_2} \cdots \tau_{a_n}\>_g^\Theta= \int_{\overline{\mathcal{M}}_{g,n}}\Theta_{g,n} \psi_1^{a_1} \psi_2^{a_2} \cdots \psi_n^{a_n}. 
\ee
Again, the integral on the right-hand side vanishes unless the stability condition~\eqref{stability} is satisfied, all $a_i$ are non-negative integers, and the dimension constraint 
\be\label{d2}
g-1=\sum_{i=1}^n a_i
\ee 
holds true.
Consider the generating function of the 
intersection numbers of $\Theta$-classes and $\psi$-classes
\be
F^\Theta_{g,n}= \sum_{a_1,\ldots,a_n\ge 0}  \frac{\prod T_{a_i}}{n!} \int_{\overline{\mathcal{M}}_{g,n}}\Theta_{g,n} \psi_1^{a_1} \psi_2^{a_2} \cdots \psi_n^{a_n}, 
\ee
then it was conjectured by Norbury \cite{Norb}, that it provides a direct analog of the Kontsevich--Witten tau-function. Recently this conjecture was proven by Chidambaram, Garcia-Failde, and Giacchetto.
\begin{theorem*}[\cite{NorbP}]
The generating function
\be
\tau_\Theta = \exp\left(\sum_{g=0}^\infty \sum_{n=0}^\infty \hbar^{2g-2+n}F_{g,n}^\Theta \right)
\ee
becomes a tau-function of the KdV hierarchy after the change of variables~$T_n=(2n+1)!!t_{2n+1}$.
\end{theorem*}
Moreover, $\tau_\Theta$ is nothing but a tau-function of the Br\'ezin--Gross--Witten (BGW) model 
\be
\tau_\Theta=\tau_{BGW}.
\ee
We refer the reader to \cite{Norb,N2,NorbP} for a detailed presentation. 

The BGW tau-function can be described by a matrix model. This matrix model description has a natural one-parameter deformation, which we will consider below. Let $\Lambda=\diag(\lambda_1,\lambda_2,\dots,\lambda_M)$ be a diagonal matrix.
For any function $f$, dependent on the infinite set of variables ${\bf t}=(t_1,t_3,t_5,\dots)$, let
\be\label{Miwa}
f\left(\left[\Lambda^{-1}\right]\right):=f({\bf t})\Big|_{t_k=\frac{1}{k}\Tr \Lambda^{-k}}
\ee
be the {\em Miwa parametrization}. The {\em Generalized Br\'ezin--Gross--Witten model} was introduced in \cite{MMS} and further investigated in \cite{ABGW}. In the Miwa parametrization it is given by the asymptotic expansion of the matrix integral
\be
\tau_{BGW}([\Lambda^{-1}],N):=\tilde{\mathcal C}^{-1} \int[ d\Phi] \exp\left(-\frac{1}{2\hbar}\Tr\left(\Lambda^2 \Phi+\Phi^{-1}+2\hbar(N-M)\log \Phi\right)\right).
\ee
where one integrates over the normal $M\times M$ matrices.
\begin{remark}
Let us note that to make the normalization compatible with the one of Norbury, we substitute $\hbar$ of \cite{ABGW} with $2\hbar$.
\end{remark}
For $N=0$ the generalized BGW tau-function reduces to the original BGW model, $\tau_{BGW}({\bf t},0)=\tau_{BGW}({\bf t})$.
For arbitrary $N$, this is a tau-function of the KdV hierarchy. Moreover, the tau-functions for different values of the parameter $N$ (not to be confused with $M$, the size of the matrices) are related to each other by the MKP hierarchy. The first terms of the series expansion of $\tau_{BGW}({\bf t},N)$ were described in \cite{ABGW}.
The Virasoro constrains for the generalized BGW model can be derived with the help of the Kac--Schwarz approach \cite{ABGW}. These Virasoro constraints lead to the cut-and-join description
\be\label{cajg}
\tau_{BGW}({\bf t},N)=e^{\hbar \widehat{W}_0(N)}\cdot 1,
\ee
with the cut-and-join operator 
\be\label{CAJN}
\widehat{W}_0(N)=\sum_{k,m\in {\mathbb Z}_{\odd}^+} \left(kmt_{k}t_{m}\frac{\p}{\p t_{k+m-1}}+\frac{1}{2}(k+m+1)t_{k+m+1}\frac{\p^2}{\p t_k \p t_m}\right)+\left(\frac{1}{8}-\frac{N^2}{2}\right)t_1.
\ee

\subsection{Cut-and-join operator and  BKP tau-function}

In \cite{ABKP} it was proven that after a transformation $t_k \mapsto t_k/2$ the cut-and-join operator (\ref{CAJN}) acquires the form
\be
\widehat{W}_0(N)=\frac{1}{4}\widehat{M}_{-1}^B+\left(\frac{1}{16}-\frac{N^2}{4}\right)\widehat{J}^B_{-1},
\ee
where the generators $\widehat{J}^B_{-1}$ and $\widehat{M}^B_{-1}$ are given by (\ref{JB}) and (\ref{MB}).
Hence, it is an element of the BKP symmetry algebra, and $\tau_{BGW}({\bf t}/2,N)$ is a tau-function of the BKP hierarchy.
\begin{remark}
The BKP integrability also immediately follows from the KdV one, see Theorem \ref{MT}.
\end{remark}
In the basis (\ref{Wop}) the cut-and-join operator has the form
\be\label{gcaj}
\widehat{W}_0(N)=\frac{1}{4}\widehat{W}_{-1,2}^B+\left(\frac{1}{16}-\frac{N^2}{4}\right)\widehat{W}_{-1,0}^B,
\ee
and is equal to $\widehat{W}^B_{{\mathtt w}_0(N)}$ with
\be
{\mathtt w}_0(N)=-\frac{1}{2}\left(z\p_z^2+\left(\frac{1}{4}-N^2\right)z^{-1}\right) \in w_{1+\infty}^B.
\ee 
In this section we will find another expression for the BKP tau-function $\tau_{BGW}({\bf t}/2,N)$ in terms of the group elements of the $W_{1+\infty}^B$ algebra.

Let 
\be
F(z;a):=\sum_{k=1}^\infty \frac{B_{k+1}(a+1/2)}{(k+1)kz^k},
\ee
where $B_{k}(z)$ are the {\em Bernoulli polynomials}.
This series appears in Stirling's expansion of the gamma function,
\be\label{Gammaas}
\Gamma(z+1/2-a)\sim \sqrt{2\pi}z^{z-a}e^{-z}e^{F(z;a)},
\ee
valid for large values of $|z|$ with $|\arg(z)|<\pi$. Here $a$ is an arbitrary finite complex number. 
Let us consider the group operator of the algebra $w_{1+\infty}^B$
\be
{\mathtt O}_B(N)=\exp\left(-\log(-N^2)D+2\sum_{k\in  {\mathbb Z}_{\even}^+}\frac{B_{k+1}(D+1/2)}{(k+1)kN^k}\right),
\ee
where
\be
D:=z\p_z.
\ee
It can be obtained from the asymptotic expansion of a combination of the gamma-functions as $N$ goes to infinity with finite $D$,
\be
{\mathtt O}_B(N)\sim \frac{\Gamma(1/2+N)}{\Gamma(D+1/2+N)}\frac{\Gamma(1/2-N)}{\Gamma(D+1/2-N)}.
\ee
It is easy to see that 
\be
D^k \in w_{1+\infty}^B, \,\,\,\,\,\,  \mathrm{for}  \quad k \in  {\mathbb Z}_{\odd}^+.
\ee
Since $B_{k}(x+1/2)\in x\,{\mathbb Q}[x^2]$ for odd $k$, the operator $\log(\mathtt{O}_B(N))$ is an odd series in $D$, therefore,
\be
\log({\mathtt O}_B(N))\in w_{1+\infty}^B.
\ee
The operator $\mathtt{w}_0(N)$ can be related to the operator $- z^{-1}/2$ by a conjugation:

\begin{lemma}\label{lem5}
\be
{\mathtt w}_0(N)=- {\mathtt O}_B(N)\,  \frac{z^{-1}}{2} \, {\mathtt O}_B(N)^{-1}.
\ee
\end{lemma}
\begin{proof}
One can rewrite the operator ${\mathtt w}_0(N)$ in terms of $D$ and $z^{-1}$,
\be
{\mathtt w}_0(N)=-\frac{z^{-1}}{2}\left(D^2-D+\frac{1}{4}-N^2\right).
\ee
Then
\begin{equation}
\begin{split}
 {\mathtt O}_B(N)\, \frac{z^{-1}}{2} \, {\mathtt O}_B(N)^{-1}
 &= \frac{z^{-1}}{2}  \frac{\Gamma(D+1/2+N)}{\Gamma(D-1/2+N)}\frac{\Gamma(D+1/2-N)}{\Gamma(D-1/2-N)}\\
&= \frac{z^{-1}}{2} (D-1/2+N)(D-1/2-N)\\
&=  \frac{z^{-1}}{2} \left(D^2-D+\frac{1}{4}-N^2\right).
\end{split}
\end{equation}
\end{proof}

Let us consider the group element of the $W_{1+\infty}^B$ algebra
\be\label{O}
\widehat{O}_B(N):=e^{\widehat{W}^B_{\log({\mathtt O}_B(N))}}.
\ee
It is obvious that $ \widehat{O}_B(N)\cdot 1=1$, then from Lemma \ref{lem5} and cut-and-join representation (\ref{cajg}) we have
\begin{theorem}\label{T1}
Tau-function of the generalized BGW model is a solution of the BKP hierarchy, given by
\be
\tau_{BGW}({\bf t}/2,N)=\widehat{O}_B(N) \cdot e^{\frac{\hbar t_1}{4}}.
\ee
\end{theorem}

Using Lemma \ref{L41} we can get the
basis vectors for the point of the BKP Sato Grassmannian $\Gr^{(0)}_{B+}$, associated with $\tau_{BGW}({\bf t}/2,N)$. 
For the basis vectors normalized by $\Phi_k=z^{k-1}(1+O(z^{-1}))$ one has
\begin{equation}
\begin{split}
\Phi_k&= \frac{\Gamma(k-1/2+N)}{\Gamma(D+1/2+N)}\frac{\Gamma(k-1/2-N)}{\Gamma(D+1/2-N)} \cdot e^{- \frac{\hbar}{2z}} z^{k-1}\\
&= \frac{\Gamma(k-1/2+N)}{\Gamma(D+1/2+N)}\frac{\Gamma(k-1/2-N)}{\Gamma(D+1/2-N)} \cdot \sum_{m=0}^\infty (-\hbar)^m \frac{z^{k-m-1}}{2^m m!}\\
&=\sum_{m=0}^\infty (-\hbar)^m  \frac{\Gamma(k-1/2+N)}{\Gamma(k-m-1/2+N)}\frac{\Gamma(k-1/2-N)}{\Gamma(k-m-1/2-N)} \frac{z^{k-m-1}}{2^m m!}. 
\end{split}
\end{equation}

\begin{remark}
Expression for $\Phi_1$ coincides with the expression for the first basis vector
of the generalized BGW tau-function, obtained within the framework of KP in \cite{ABGW} see formula (77), if one substitutes $\hbar$ with $2\hbar$ and $\lambda$ with $-z$. However, the points of the Sato Grassmannian associated with KP and BKP frameworks are different, see Section \ref{kdvb}.
\end{remark}


\subsection{Schur Q-functions and BKP hierarchy}\label{S2}

Schur Q-functions were introduced by Schur \cite{SchQ} for the description of the projective representations of the symmetric groups. 
These functions are labeled by  strict partitions $\lambda \in \DP$. The definition and description of the Schur Q-functions can be found in Section 3.8 of Macdonald's book \cite{Mac}. 
In many aspects the Schur Q-functions are similar to the ordinary Schur functions.
For example, let us mention the {\em Cauchy formula}
\be\label{Cauchy}
 \sum_{\lambda \in \DP}  \frac{Q_\lambda({\bf t}) Q_\lambda({\bf s})}{2^{\ell(\lambda)}}=\exp \left(2 \sum_{k\in {\mathbb Z}_{\odd}^+} k t_{k} {s}_{k}\right),
\ee
and an analog of the standard hook formula \cite{SchQ}
\be\label{hook}
Q_\lambda(\delta_{k,1})=2^{|\lambda|}\frac{1}{\prod_{j=1}^{\ell(\lambda)} \lambda_j!} \prod_{k<m}\frac{\lambda_k-\lambda_m}{\lambda_k+\lambda_m}.
\ee

Schur Q-functions are closely related to the BKP hierarchy. Let us consider the class of hypergeometric solutions of the BKP hierarchy introduced by Orlov \cite{OBKP}. It was shown by Mironov, Morozov, and Natanzon \cite{MMNQ} that the hypergeometric solutions can be interpreted as the generating functions of the spin Hurwitz numbers. Hypergeometric solutions of the 2-component BKP hierarchy are given by the following sums over strict partitions
\begin{equation}
\begin{split}\label{HG}
\tau({\bf t},{\bf s})=\sum_{\lambda \in \DP} 2^{-\ell(\lambda)}r_\lambda Q_\lambda({\bf t}/2)Q_\lambda({\bf s}/2),
\end{split}
\end{equation}
where
\be
r_\lambda=\prod_{j=1}^{\ell(\lambda)} r(1)r(2)\dots r(\lambda_j)
\ee
for some function $r(z)$, which we call the {\em weight generating function}.

\begin{remark}
Let us note that if $r(k)=0$ for some $k \in {\mathbb Z}^+$, then the sum in (\ref{HG}) is finite and contains only partitions with $\lambda_1<k$. 
\end{remark}

Relation between Schur Q-functions and BKP hierarchy, in particular, is described by the following result of You:
\begin{theorem}[\cite{You}]\label{TYou}
For the states (\ref{lambda}) the boson-fermion correspondence yields
\be
\sigma_{B}^0 (\left| \lambda \right> )= 2^{-\ell(\lambda)/2} Q_\lambda({\bf t}/2).
\ee
Schur Q-functions are the polynomial solutions of the BKP hierarchy.
\end{theorem}

\subsection{Schur Q-function expansion of the generalized BGW tau-function}\label{S3}

In \cite{ABKP} we suggest the following:
\begin{conjecture}\label{MC}
\be\label{MF}
\tau_{BGW}({\bf t})= \sum_{\lambda \in \DP} \left(\frac{\hbar}{16}\right)^{|\lambda|} \frac{Q_\lambda({\bf t}) Q_\lambda(\delta_{k,1})^3}{2^{\ell(\lambda)} Q_{2\lambda}(\delta_{k,1})^2}.
\ee
\end{conjecture}

For the generalized BGW model with arbitrary $N$ in \cite{ABKP} we conjecture a Schur Q-function expansion, which generalizes (\ref{MF}). This expansion is of the hypergeometric type (\ref{HG}), and the author conjectured in \cite{ABKP} that the generalized BGW tau-function, after a simple rescaling of times, is a hypergeometric solution of the BKP hierarchy. We prove this conjecture below, see Theorem \ref{T3}, and this is one of the most important results of this paper. 

It was also noted in \cite{ABKP}, see Remark 3.1, that the analogous to (\ref{MF}) Schur Q-function expansion for the Kontsevich--Witten tau-function was established in a series of papers \cite{DFIZ,JF,MMQ}. In \cite{ABKP} we conjecture that this tau-function is also a hypergeometric solution of the BKP hierarchy. 
The statement of this conjecture follows from the Schur Q-function expansion of \cite{DFIZ,JF,MMQ} and the relations for the values of the Schur Q-function on special loci proven in \cite{MMNO} (see also \cite{Liu}).

Let us use Theorem \ref{TYou} to describe the Schur Q-function expansion of the generalized BGW model. 
From (\ref{Wferm1}) we have
\be
\widehat{W}_{D^k}^B\cdot  \lvac e^{J_+({\bf t}) }=\lvac e^{J_+({\bf t}) }\frac{1}{2}\sum_{m\in {\mathbb Z}}(-1)^m m^k\Normord{\phi_{-m} \phi_m}.
\ee
Both sides of this identity vanish for even $k$. For the group operator we have
\be
\widehat{O}_B(N) \cdot  \lvac e^{J_+({\bf t}) } =   \lvac e^{J_+({\bf t}) } O (N),
\ee
where
\be\label{fermoo}
O (N) = \exp\left( \sum_{m\in {\mathbb Z}}(-1)^m \left(-\frac{1}{2}\log(-N^2)m+\sum_{k\in  {\mathbb Z}_{\even}^+}\frac{B_{k+1}(m+1/2)}{(k+1)kN^k}\right)\Normord{\phi_{-m} \phi_m}\right).
\ee
The tau-function is given by the vacuum expectation value
\be
\tau_{BGW}({\bf t}/2,N)=\lvac e^{J_+({\bf t}) }  O (N)  e^{\frac{\hbar}{2} J_{-1}}\rvac. 
\ee

Operators $\widehat{W}_{D^k}^B$ are diagonal operators of the $\go(\infty)$ algebra.
For odd $k$ using anticommutation relations for (\ref{lambda}) we have
\be
\frac{1}{2}\sum_{m\in {\mathbb Z}}(-1)^m m^k\Normord{\phi_{-m} \phi_m} \left| \lambda \right> = - \sum_{j=1}^{\ell(\lambda)} \lambda_j^k    \left| \lambda \right>,
\ee
hence, from Theorem \ref{TYou} for odd $k$ one has
\be\label{af}
\widehat{W}_{D^k}^B\cdot Q_\lambda({\bf t}/2) = - \sum_{j=1}^{\ell(\lambda)} \lambda_j^k \, Q_\lambda({\bf t}/2). 
\ee

Let us introduce
\be\label{weghtbgw}
r^{(N)}(z)=\frac{(2z-1)^2-4N^2}{4}.
\ee
Theorem \ref{T1} leads to the proof of Conjecture 2 of  \cite{ABKP}:
\begin{theorem}\label{T3}
The partition function of the generalized BGW model in the properly normalized times, $\tau_{BGW}({\bf t}/2, N)$ is a hypergeometric tau-function of the BKP hierarchy (\ref{HG}) with   $s_k^*=2\delta_{k,1}$:
\be\label{thex}
\tau_{BGW}({\bf t}/2,N)= \sum_{\lambda \in \DP} \left(\frac{\hbar}{4}\right)^{|\lambda|}  r_\lambda^{(N)} \frac{Q_\lambda({\bf t}/2) Q_\lambda(\delta_{k,1})}{ 2^{\ell(\lambda)}}.
\ee
\end{theorem}
\begin{proof}
Using Cauchy formula (\ref{Cauchy}) one gets
\be
e^{\frac{\hbar t_1}{4}}= \sum_{\lambda \in \DP} \left(\frac{\hbar}{4}\right)^{|\lambda|}  \frac{Q_\lambda({\bf t}/2) Q_\lambda(\delta_{k,1})}{2^{\ell(\lambda)}}.
\ee
From Theorem \ref{T1} and formula (\ref{af}) we have
\begin{equation}
\begin{split}
\tau_{BGW}({\bf t}/2,N)&=\widehat{O}_B(N)\cdot \sum_{\lambda \in \DP} \left(\frac{\hbar}{4}\right)^{|\lambda|}  \frac{Q_\lambda({\bf t}/2) Q_\lambda(\delta_{k,1})}{2^{\ell(\lambda)}}\\
&= \sum_{\lambda \in \DP} \left(\frac{\hbar}{4}\right)^{|\lambda|}  \frac{Q_\lambda({\bf t}/2) Q_\lambda(\delta_{k,1})}{2^{\ell(\lambda)}}  \prod_{j=1}^{\ell(\lambda)} \frac{\Gamma(1/2+N)}{\Gamma(-\lambda_j+1/2+N)}\frac{\Gamma(1/2-N)}{\Gamma(-\lambda_j+1/2-N)}\\
&= \sum_{\lambda \in \DP} \left(\frac{\hbar}{4}\right)^{|\lambda|}  \frac{Q_\lambda({\bf t}/2) Q_\lambda(\delta_{k,1})}{2^{\ell(\lambda)}}  \prod_{j=1}^{\ell(\lambda)} \frac{\Gamma(\lambda_j+1/2+N)}{\Gamma(1/2+N)}\frac{\Gamma(\lambda_j+1/2-N)}{\Gamma(1/2-N)}.
\end{split}
\end{equation}
This formula is equivalent to that of Theorem.
\end{proof}
Please note that all $\lambda_j$ are nonnegative, therefore no negative powers of $N$ appear in the proof.

The first few terms of the expansion \eqref{thex} are given by
\begin{equation}
\begin{split}
\tau_{BGW}({\bf t}/2,N)=1+\frac{\hbar}{16}(1-4N^2)Q_{(1)}({\bf t}/2)+\frac{\hbar^2}{256}(1-4N^2)(9-4N^2)Q_{(2)}({\bf t}/2)\\
+\frac{\hbar^3}{12288}(1-4N^2)(9-4N^2)(2(25-4N^2)Q_{(3)}({\bf t}/2)+(1-4N^2)Q_{(2,1)}({\bf t}/2))+O(\hbar^4).
\end{split}
\end{equation}

Let us stress that for the BGW tau-function ($N=0$) the statement of this theorem reduces to equation (\ref{MF}), because from (\ref{hook}) we have
\be
\frac{Q_\lambda(\delta_{k,1})}{Q_{2\lambda}(\delta_{k,1})}=\prod_{j=1}^{\ell(\lambda)} (2\lambda_j-1)!!.
\ee

For $N=k+1/2$ with $k\in {\mathbb Z}_{\geq 0}$  the weight generating function (\ref{weghtbgw}) vanishes for $z=k+1$, $r^{(k+1/2)}(k+1)=0$, hence, in this case the tau-function $\tau_{BGW}$ is a polynomial. These polynomial tau-functions are given by certain shifted Schur functions for the triangular partitions \cite{ABGW}, which coincide with the shifted Schur Q-functions \cite{ABKP}.

The formula in Theorem \ref{T3} is a direct analog of the Schur function expansion of the generating function for the strictly monotone Hurwitz numbers or Grothendieck's dessins d'enfant.
We claim that the generalized BGW model is a generating function of the spin version of the disconnected strictly monotone Hurwitz numbers. In this context the cut-and-join description of the generalized BGW model is a direct analog of the Zograf cut-and-join description of Grothendieck's dessins d'enfant \cite{Zog}. We will discuss the details of this identification with the spin Hurwitz numbers elsewhere.


\section{Higher models}\label{S6}

In this section we discuss a family of hypergeometric tau-functions for the  2-component BKP hierarchy. This family includes the model, considered in the previous section, and can serve as a deep generalization of this model.

 First, let us note that the operator $\widehat{O}_B(N)$ is a group operator, therefore it is invertible. Consider
\be\label{td}
\tau_{a,b}({\bf t},{\bf s},{\bf u}, {\bf w})=\frac{ \widehat{O}_B(u_a) \widehat{O}_B(u_{a-1}) \dots  \widehat{O}_B(u_1)}{\widehat{O}_B(w_b) \widehat{O}_B(w_{b-1})\dots\widehat{O}_B(w_1)} \cdot 
e^{\frac{1}{2} \sum_{k\in {\mathbb Z}_{\odd}^+}  k t_{k} {s}_{k}},
\ee
where $u_k$ and $w_k$ are formal parameters. We will focus on the case when the number of the parameters $u_k$ and $w_k$ is finite, however, many formulas below are valid also for the general case with infinitely many parameters.
We regard (\ref{td}) as a formal power series in the variables ${\bf t}$ and ${\bf s}$, and parameters ${\bf u}$ and ${\bf w}$.
By construction, this is a tau-function of 2-component BKP hierarchy. The definition immediately implies the free fermion formula
\be
\tau_{a,b}({\bf t},{\bf s},{\bf u},{\bf w})=\lvac e^{J_+({\bf t})} O(u_a)\dots O(u_1)  O^{-1}(w_b)\dots  O^{-1}(w_1)e^{J_-({\bf s})}  \rvac,
\ee
where the diagonal operators $O$ are given by (\ref{fermoo}).

For $a=1$, $b=0$ and $s_j=\delta_{j,1}\hbar/2$ it reduces to the generalized BGW tau-function, considered above
\be
\tau_{1,0}({\bf t},\delta_{k,1}\hbar/2, N,0)= \tau_{BGW}({\bf t}/2,N).
\ee
In (\ref{td}) we do not introduce the expansion parameter $\hbar$, it can be restored by rescaling of the variables ${\bf t}$, ${\bf s}$ and parameters ${\bf u}$, ${\bf w}$.

Operators $\widehat{O}_B(N_m)$ commute with each other, hence, their order in (\ref{td}) is not important
\begin{lemma}
For all $u_1,u_2$
\be
\left[\widehat{O}_B(u_1),\widehat{O}_B(u_2))\right]=0.
\ee
\end{lemma}
\begin{proof}
From (\ref{O}) we see that $\log(\widehat{O}_B(N))$ is a linear combination of $\widehat{W}^B_{D^k}$ with $k \in {\mathbb Z}_{\odd}^+$. Lemma \ref{Wcomm} implies the identity
\be
\left[\widehat{W}^B_{D^k},\widehat{W}^B_{D^m}\right]=0
\ee
for all $k,m\in {\mathbb Z}_{\odd}^+$, from which the statement of the lemma immediately follows.
\end{proof}

By analogy with Theorem \ref{T3} for the weight generating function, given by (\ref{weghtbgw}), we have
\begin{proposition}
\be\label{generex}
\tau_{a,b}({\bf t},{\bf s},{\bf u}, {\bf w})= \sum_{\lambda \in \DP} \frac{r_\lambda^{(u_1)}  r_\lambda^{(u_2)} \dots r_\lambda^{(u_a)} }{r_\lambda^{(w_1)}  r_\lambda^{(w_2)} \dots r_\lambda^{(w_b)}}   \frac{Q_\lambda({\bf t}/2) Q_\lambda({\bf s}/2) }{ 2^{-\ell(\lambda)}  }
\ee
\end{proposition}

\begin{remark}
It is easy to see that in general (\ref{td}) is not a tau-function of the KdV hierarchy.
\end{remark}

We claim that these tau-functions are the generating functions for the interesting family of the weighted spin Hurwitz numbers, which are the spin analogs to the
double weighted Hurwitz numbers with rational weight generating functions, see, e.g., \cite{ACEH} and references therein. This section contains description of this family of the BKP tau-functions analogous to the description of the corresponding family for the ordinary Hurwitz numbers obtained in \cite{AN,ALS}. Possible choices of the weights for the weighed spin Hurwitz numbers associated with the  hypergeometric tau-functions for the  2-component BKP hierarchy are discussed in \cite{AS}.

\begin{remark}
This family of BKP tau-functions for $s_k=\delta_{k,1}$ was introduced by Orlov 
\cite{OBKP}. He proved, in particular, that this family is related to the hypergeometric solutions of the KP hierarchy by a square root relation.
Therefore, the generating functions of the spin Hurwitz numbers (\ref{td}) are related to the corresponding generating functions of the ordinary Hurwitz numbers by a square root relation, completely analogous to the one investigated in \cite{root}.
\end{remark}

\begin{remark}
It would be interesting to construct and investigate the matrix model description of the tau-functions (\ref{td}) similar to the one, constructed in \cite{AC,AN} for the hypergeometric solutions of the KP hierarchy.
\end{remark}


\subsection{Cut-and-join description}

Let us derive the cut-and-join description for the tau-functions (\ref{td}). 
Consider the operator 
\be
{\mathtt w}({\bf u},{\bf w})= z^{-1}\frac{ \prod_{j=1}^a \left((D-1/2)^2-u_j^2\right)}{\prod_{j=1}^b \left((D-1/2)^2-w_j^2\right)} \in w_{1+\infty}^B. 
\ee

\begin{theorem}\label{propcaj}
\be
\tau_{a,b}({\bf t},{\bf s},{\bf u}, {\bf w})=e^{- \sum_{k\in {\mathbb Z}_{\odd}^+}   {s}_{k}  \widehat{W}_{{\mathtt w}({\mathbf u},{\mathbf w})^k} }\cdot 1
\ee
\end{theorem}
\begin{proof}
Comparing (\ref{wmin}), (\ref{JB}) and (\ref{Wop}) we have
\be
k t_k = -2 \widehat{W}_{z^{-k}} \,\,\,\,\, \mathrm{for}  \quad  k \in  {\mathbb Z}_{\odd}^+.
\ee
Therefore
\be
\tau_{a,b}({\bf t},{\bf s},{\bf u}, {\bf w})=e^{\sum_{j=1}^a \widehat{W}^B_{ \log( {\mathtt O}_B(u_j))}-\sum_{j=1}^b \widehat{W}^B_{ \log( {\mathtt O}_B(w_j))}}  \cdot e^{-\sum_{k\in {\mathbb Z}_{\odd}^+}    {s}_{k}  \widehat{W}_{z^{-k}} }.
\ee
Repeating the proof of Lemma  \ref{lem5} we have
\be
\frac{ \prod_{j=1}^a {\mathtt O}_B(u_j)}{ \prod_{j=1}^b {\mathtt O}_B(w_j)}\, z^{-1} \,  \frac{\prod_{j=1}^b {\mathtt O}_B(w_j)}{\prod_{j=1}^a {\mathtt O}_B(u_j)} = {\mathtt w}({\mathbf u}, {\mathbf w}),
  \ee
and the statement of the theorem follows from Lemma \ref{Wcomm}.
\end{proof}

\begin{remark}
From Lemma \ref{Wcomm} it follows that the operators $ \widehat{W}_{{\mathtt w}({\bf u})^k}$ commute with each other
\be
\left[ \widehat{W}_{{\mathtt w}({\bf u},{\bf w})^k}, \widehat{W}_{{\mathtt w}({\bf u},{\bf w})^m}\right]=0.
\ee
thus, from Theorem \ref{propcaj} we have the equation
\be
\frac{\p}{\p s_k}\tau_{a,b}({\bf t},{\bf s},{\bf u}, {\bf w})= -  \widehat{W}_{{\mathtt w}({\mathbf u},{\mathbf w})^k} \cdot  \tau_{a,b}({\bf t},{\bf s},{\bf u}, {\bf w}).
\ee
\end{remark}

The case $b=0$ leads to the simple group operators. We expect, that, while the case $b\neq 0$ leads to the more complicated group operators, it still can be described by interesting cut-and-join equations.

By construction, the cut-and-join operator belongs to the the symmetry algebra of the BKP hierarchy. 
If only a finite number of the variables $s_k$ do not vanish, then the sum $ \sum_{k\in {\mathbb Z}_{\odd}^+}   {s}_{k}  \widehat{W}_{{\mathtt w}({\bf u})^k}$ is finite, and for the finite number of parameters ${\bf u}$ the operator $ \widehat{W}_{{\mathtt w}({\bf u})^k}$ is a finite degree differential operator in variables ${\bf t}$. 

Consider the case $s_j=\delta_{j,1}\hbar/2$. It corresponds to the single weighted spin Hurwitz numbers. Then the sum over partitions (\ref{generex}) can be expanded as
\be
\tau_{a,0}({\bf t},\delta_{j,1}\hbar/2,{\bf u},{\bf 0})=\sum_{k=0}^\infty \hbar^k \tau_{a,0}^{(k)}({\bf t},{\bf u}),
\ee
where
\be
\tau_{a,0}^{(k)}({\bf t},{\bf u})= \sum_{\lambda \in \DP, |\lambda|=k} r_\lambda^{(u_1)}  r_\lambda^{(u_2)} \dots r_\lambda^{(u_a)}    \frac{Q_\lambda({\bf t}/2) Q_\lambda(\delta_{j,1}) }{ 4^{|\lambda|}2^{-\ell(\lambda)}  }.
\ee
Then the cut-and-join operator provides the recursion in degree of the ramified covering:
\be
\tau_{a,0}^{(k)}({\bf t},{\bf u}) = -\frac{1}{2k} \widehat{W}_{{\mathtt w}({\bf u})}\cdot \tau_{a,0}^{(k-1)}({\bf t},{\bf u}).
\ee
We have a similar expansion for $s_j=\delta_{j,k}\hbar/2$ with $k>1$, it corresponds to the case of orbifold spin Hurwitz numbers. 

For the simplest case $a=1$,  $s_j=\delta_{j,1}\hbar/2$ and $u_1=N$ this operator reduces the the cut-and-join operator (\ref{gcaj}) for the generalized BGW tau-function.


\subsection{Sato Grassmannian and Kac--Schwarz operators}
Let
\be
{\mathtt G}=\frac{\prod_{j=1}^a {\mathtt O}_B(u_j)}{ \prod_{j=1}^b {\mathtt O}_B(w_j)} 
e^{- \sum_{k\in {\mathbb Z}_{\odd}^+}   {s}_{k} z^{-k}  }
\ee
be the group operators of the $w_{1+\infty}^B$ algebra.
The point of the Sato Grassmannian, corresponding to the BKP tau-function (\ref{td}) is given by
\be
{\mathcal W}= {\mathtt G} \cdot{\mathcal W}_0,
\ee
where ${\mathcal W}_0=\{1,z,z^2,z^3\dots\}$ is a point of the Grassmannian for the trivial BKP tau-function $\tau({\bf t})=1$.

The normalized basis vectors for the tau-function (\ref{td}) are given by
\be\label{Phi}
\Phi_k=  \frac{\prod_{j=1}^b\Gamma(k-1/2+w_j)\Gamma(k-1/2-w_j)}{ \prod_{j=1}^a \Gamma(k-1/2+u_j)\Gamma(k-1/2-u_j)}  {\mathtt G}\cdot z^{k-1}.
\ee

Let us introduce the operators
\begin{equation}\label{pq}
\begin{split}
{\mathtt p}&:={\mathtt G} \, \p_z\, {\mathtt G}^{-1},\\
{\mathtt q}&:={\mathtt G}\, z\, {\mathtt G}^{-1}.
\end{split}
\end{equation}
These are the Kac-Schwarz operators
\be
{\mathtt p}\cdot  {\mathcal W}\subset {\mathcal W}, \,\,\,\,\,\,\, {\mathtt q}\cdot  {\mathcal W}\subset {\mathcal W},
\ee
 satisfying the canonical commutation relation
\be
\left[{\mathtt p},{\mathtt q}\right]={\mathtt 1}.
\ee

Let us introduce
\be
{\mathtt R}_m=\frac{ \prod_{j=1}^b ((D+1/2-m)^2-w_j^2)}{\prod_{j=1}^a ((D+1/2-m)^2-u_j^2)}.
\ee
For operators (\ref{pq}) we have
\begin{equation}
\begin{split}
{\mathtt q}&= \frac{\prod_{j=1}^a {\mathtt O}_B(u_j)}{ \prod_{j=1}^b {\mathtt O}_B(w_j)}  \,z \,
\frac{ \prod_{j=1}^b {\mathtt O}_B(w_j)}{\prod_{j=1}^a {\mathtt O}_B(u_j)} \\
 &=z \, {\mathtt R}_0
\end{split}
\end{equation}
and
\begin{equation}
\begin{split}
{\mathtt p}&=\frac{\prod_{j=1}^a((D+1/2)^2-u_j^2)}{\prod_{j=1}^b((D+1/2)^2-w_j^2)}\p_z- \sum_{k\in {\mathbb Z}_{\odd}^+} k  {s}_{k} {\mathtt q}^{-k-1}\\
&={\mathtt R}_0^{-1}\p_z- \sum_{k\in {\mathbb Z}_{\odd}^+} k  {s}_{k} z^{-k-1} \prod_{m=1}^{k+1} {\mathtt R}_m^{-1}.
\end{split}
\end{equation}

By construction for the wave function $\Psi=\Phi_1$ from (\ref{Phi}) we have
\be
{\mathtt p}\cdot \Psi(z) =0.
\ee
However, it is more convenient to introduce the operator 
\be
{\mathtt c}={\mathtt q}{\mathtt p}=D- \sum_{k\in {\mathbb Z}_{\odd}^+} k  {s}_{k} {\mathtt q}^{-k}
\ee
which also annihilates the wave function, ${\mathtt c}\cdot \Psi=0$.
Let us consider the case when $s_\ell \neq 0$ for some finite odd $\ell $ and $s_k=0$ for $k>\ell$. If $b=0$, then the operator ${\mathtt c}$ is a finite order differential operator with coefficients in ${\mathbb C}[z,z^{-1}]$. If $b>0$ we introduce
\be
{\mathtt K}=\prod_{r=1}^\ell\prod_{j=1}^b ((D+r-1/2)^2-w_j^2).
\ee
Then for the operator ${\mathtt A}:={\mathtt K}  {\mathtt c}$ we have
\begin{equation}
\begin{split}\label{qsc}
{\mathtt A}=D{\mathtt K}- \sum_{k=1}^\ell k  {s}_{k}   z^{-k}  \prod_{r=1}^{\ell-k}\prod_{j=1}^b ((D+r-1/2)^2-w_j^2)\prod_{m=0}^{k-1}  \prod_{j=1}^a ((D-m-1/2)^2-u_j^2).
\end{split}
\end{equation}
It is a polynomial of $z^{-1}$ and $D$, which annihilates the wave function
\be
{\mathtt A}\cdot \Psi(z)=0.
\ee
We can interpret it as the {\em quantum spectral curve}. For ${\mathtt K}\neq 1$ this operator, in general, is not a Kac-Schwarz operator.

\begin{conjecture}
The $n$-point functions  for the tau-function (\ref{td}) can be given by the topological recursion, associated with the semi-classical limit of the quantum spectral curve (\ref{qsc}).
\end{conjecture}


\subsection{Examples}

The case of the generalized BGW tau-function, considered in Section \ref{GBGW}, corresponds to
$a=1$, $b=0$,  $s_j=\delta_{j,1}\hbar/2$. In this case
\be
{\mathtt R}_m=\frac{1}{ (D+1/2-m)^2-u^2}.
\ee
The Kac--Schwarz operators, introduced in the previous section, are
\begin{equation}
\begin{split}
{\mathtt p}&=( (D+1/2)^2-u^2)\p_z-\frac{\hbar}{2} z^{-2}  {\mathtt R}_1^{-1} {\mathtt R}_2^{-1},\\
{\mathtt q}&=z\frac{1}{ (D+1/2)^2-u^2},
\end{split}
\end{equation}
the operator
\be\label{cgBGW}
{\mathtt c}=D-\frac{\hbar}{2}z^{-1} ( (D-1/2)^2-u^2)
\ee
coincides with the quantum spectral curve operator ${\mathtt A}$. The quantum spectral curve equation in this case is closely related to the Bessel equation.

We expect, that the spin Hurwitz numbers interpretation of the case with $a=0$, $b=1$ is particularly interesting. Namely, this case can be associated with a spin version of the monotone Hurwitz numbers. Let us consider it in more detail for $s_j=\delta_{j,1}\hbar/2$. Then
\be
{\mathtt R}_m=  (D+1/2-m)^2-w^2,
\ee
and the Kac--Schwarz operators are
\begin{equation}
\begin{split}
{\mathtt p}&=\frac{1}{ (D+1/2)^2-w^2}\p_z-\frac{\hbar}{2} z^{-2} {\mathtt R}_1^{-1} {\mathtt R}_2^{-1},\\
{\mathtt q}&=z ((D+1/2)^2-w^2),
\end{split}
\end{equation}
and
\be
{\mathtt c}=D-\frac{\hbar}{2}z^{-1} \frac{1}{(D-1/2)^2-w^2}.
\ee
Then the quantum spectral curve is given by the third order differential operator
\be
{\mathtt A}=((D-1/2)^2-w^2)D-\frac{\hbar}{2}z^{-1}.
\ee


\subsection{Constraints}

We expect that the tau-functions (\ref{td}) satisfy the Virasoro and W-constraints, analogous to ones obtained in \cite{KZ,AN}. Some of these constraints can be derived using the Kac-Schwarz approach for BKP hierarchy. However, not all of them can be derived this way. 

The illustrative example here is the generalized BGW tau-function, considered in the previous section. This model is completely described by the Virasoro constraints, derived in \cite{ABGW} using the Kac-Schwarz description of the KP hierarchy.  However, these constraints do not appear in the BKP picture. Hence, for the generalized BGW tau-function, the KP and BKP pictures are complimentary to each other: the Virasoro constraints are natural symmetries of the KP hierarchy, while the cut-and-join operator is an element of the BKP symmetry group. We expect that such interplay between different integrable structures should be common for the enumerative geometry tau-functions.

The simplest constraint for the generalzed BGW model, which can be obtained from the BKP approach, is given by
\be\label{Wop6}
\widehat{W}^B_{{\mathtt c}} \cdot \tau_{BGW}({\bf t}/2,u)=\mu \tau_{BGW}({\bf t}/2,u)
\ee
where  ${\mathtt c}$ is given by (\ref{cgBGW}), and $\mu$ is some eigenvalue. From the consideration of this equation at ${\bf t}={\bf 0}$ it is easy to see that $\mu=0$. The operator
\be
\widehat{W}^B_{{\mathtt c}}=-\frac{1}{2}\widehat{W}_{0,1}^B+ \frac{\hbar}{4}\widehat{W}_{-1,2}^B+\hbar\left(\frac{1}{16}-\frac{u^2}{4}\right)\widehat{W}_{-1,0}^B
\ee
describes the cut-and-join equation. Indeed, comparing this operator with \eqref{gcaj}, we see that \eqref{Wop6} reduces to
\be
\frac{1}{2}\widehat{W}_{0,1}^B \cdot \tau_{BGW}({\bf t}/2,u) = \hbar \widehat{W}_0(u) \cdot \tau_{BGW}({\bf t}/2,u). 
\ee
Since $\widehat{W}_{0,1}^B=\widehat{L}_0^B=2 \sum_{k \in {\mathbb Z}_{\odd}^+} k t_k\frac{\p}{\p t_k}$, we have the cut-and-join equation
\be
 \sum_{k \in {\mathbb Z}_{\odd}^+} k t_k\frac{\p}{\p t_k} \tau_{BGW}({\bf t}/2,u) = \hbar \widehat{W}_0(u) \cdot \tau_{BGW}({\bf t}/2,u)
\ee
with the unique solution normalized by $\tau_{BGW}({\bf 0},u)=1$. This solution is given by \eqref{cajg}.

\bibliographystyle{alphaurl}
\bibliography{KPTRref}

\begin{thebibliography}{MMNO21}

\bibitem[AC14]{AC}
Jan Ambj{\o}rn and Leonid Chekhov.
\newblock The matrix model for dessins d'enfants.
\newblock {\em Ann. Inst. Henri Poincar\'{e} D}, 1(3):337--361, 2014.
\newblock \href {http://dx.doi.org/10.4171/AIHPD/10}
  {\path{doi:10.4171/AIHPD/10}}.

\bibitem[ACEH20]{ACEH}
A.~Alexandrov, G.~Chapuy, B.~Eynard, and J.~Harnad.
\newblock Weighted {H}urwitz numbers and topological recursion.
\newblock {\em Comm. Math. Phys.}, 375(1):237--305, 2020.
\newblock \href {http://dx.doi.org/10.1007/s00220-020-03717-0}
  {\path{doi:10.1007/s00220-020-03717-0}}.

\bibitem[Ale18]{ABGW}
A.~Alexandrov.
\newblock Cut-and-join description of generalized {B}rezin-{G}ross-{W}itten
  model.
\newblock {\em Adv. Theor. Math. Phys.}, 22(6):1347--1399, 2018.
\newblock \href {http://dx.doi.org/10.4310/ATMP.2018.v22.n6.a1}
  {\path{doi:10.4310/ATMP.2018.v22.n6.a1}}.

\bibitem[Ale21a]{ABKP}
Alexander Alexandrov.
\newblock {Intersection numbers on {$\overline{\mathcal{M}}_{g,n}$} and {BKP}
  hierarchy}.
\newblock {\em J. High Energy Phys.}, (9):Paper No. 013, 14, 2021.
\newblock \href {http://dx.doi.org/10.1007/jhep09(2021)013}
  {\path{doi:10.1007/jhep09(2021)013}}.

\bibitem[Ale21b]{AKdV}
Alexander Alexandrov.
\newblock Kd{V} solves {BKP}.
\newblock {\em Proc. Natl. Acad. Sci. USA}, 118(25):Paper No. e2101917118, 2,
  2021.
\newblock \href {http://dx.doi.org/10.1073/pnas.2101917118}
  {\path{doi:10.1073/pnas.2101917118}}.

\bibitem[ALS16]{ALS}
A.~Alexandrov, D.~Lewanski, and S.~Shadrin.
\newblock Ramifications of {H}urwitz theory, {KP} integrability and quantum
  curves.
\newblock {\em J. High Energy Phys.}, (5):124, front matter+30, 2016.
\newblock \href {http://dx.doi.org/10.1007/JHEP05(2016)124}
  {\path{doi:10.1007/JHEP05(2016)124}}.

\bibitem[AMMN14]{AN}
A.~Alexandrov, A.~Mironov, A.~Morozov, and S.~Natanzon.
\newblock On {KP}-integrable {H}urwitz functions.
\newblock {\em J. High Energy Phys.}, (11):080, front matter+30, 2014.
\newblock \href {http://dx.doi.org/10.1007/JHEP11(2014)080}
  {\path{doi:10.1007/JHEP11(2014)080}}.

\bibitem[AS21]{AS}
Alexander Alexandrov and Sergey Shadrin.
\newblock {Elements of spin Hurwitz theory: closed algebraic formulas, blobbed
  topological recursion, and a proof of the Giacchetto-Kramer-Lewanski
  conjecture}, 2021.
\newblock \href {http://arxiv.org/abs/2105.12493} {\path{arXiv:2105.12493}}.

\bibitem[CGFG22]{NorbP}
Nitin~Kumar Chidambaram, Elba Garcia-Failde, and Alessandro Giacchetto.
\newblock {Relations on {$\overline{\mathcal{M}}_{g,n}$} and the negative
  $r$-spin Witten conjecture}, 2022.
\newblock \href {http://arxiv.org/abs/2205.15621} {\path{arXiv:2205.15621}}.

\bibitem[DFIZ93]{DFIZ}
P.~Di~Francesco, C.~Itzykson, and J.-B. Zuber.
\newblock Polynomial averages in the {K}ontsevich model.
\newblock {\em Comm. Math. Phys.}, 151(1):193--219, 1993.
\newblock \href {http://dx.doi.org/10.1007/BF02096753}
  {\path{doi:10.1007/BF02096753}}.

\bibitem[DJKM82]{JMBKP}
Etsuro Date, Michio Jimbo, Masaki Kashiwara, and Tetsuji Miwa.
\newblock Transformation groups for soliton equations. {IV}. {A} new hierarchy
  of soliton equations of {KP}-type.
\newblock {\em Phys. D}, 4(3):343--365, 1981/82.
\newblock \href {http://dx.doi.org/10.1016/0167-2789(82)90041-0}
  {\path{doi:10.1016/0167-2789(82)90041-0}}.

\bibitem[DKM81]{JMV}
Etsuro Date, Masaki Kashiwara, and Tetsuji Miwa.
\newblock Transformation groups for soliton equations. {II}. {V}ertex operators
  and {$\tau $} functions.
\newblock {\em Proc. Japan Acad. Ser. A Math. Sci.}, 57(8):387--392, 1981.
\newblock \href {http://dx.doi.org/10.3792/pjaa.57.387}
  {\path{doi:10.3792/pjaa.57.387}}.

\bibitem[EOP08]{EOP}
Alex Eskin, Andrei Okounkov, and Rahul Pandharipande.
\newblock The theta characteristic of a branched covering.
\newblock {\em Adv. Math.}, 217(3):873--888, 2008.
\newblock \href {http://dx.doi.org/10.1016/j.aim.2006.08.001}
  {\path{doi:10.1016/j.aim.2006.08.001}}.

\bibitem[Gun16]{Guni}
Sam Gunningham.
\newblock Spin {H}urwitz numbers and topological quantum field theory.
\newblock {\em Geom. Topol.}, 20(4):1859--1907, 2016.
\newblock \href {http://dx.doi.org/10.2140/gt.2016.20.1859}
  {\path{doi:10.2140/gt.2016.20.1859}}.

\bibitem[J\'95]{JF}
Tadeusz J\'{o}zefiak.
\newblock Symmetric functions in the {K}ontsevich-{W}itten intersection theory
  of the moduli space of curves.
\newblock {\em Lett. Math. Phys.}, 33(4):347--351, 1995.
\newblock \href {http://dx.doi.org/10.1007/BF00749688}
  {\path{doi:10.1007/BF00749688}}.

\bibitem[Kon92]{Kon92}
Maxim Kontsevich.
\newblock Intersection theory on the moduli space of curves and the matrix
  {A}iry function.
\newblock {\em Comm. Math. Phys.}, 147(1):1--23, 1992.
\newblock \href {http://dx.doi.org/10.1007/BF02099526}
  {\path{doi:10.1007/BF02099526}}.

\bibitem[KZ15]{KZ}
Maxim Kazarian and Peter Zograf.
\newblock Virasoro constraints and topological recursion for {G}rothendieck's
  dessin counting.
\newblock {\em Lett. Math. Phys.}, 105(8):1057--1084, 2015.
\newblock \href {http://dx.doi.org/10.1007/s11005-015-0771-0}
  {\path{doi:10.1007/s11005-015-0771-0}}.

\bibitem[Lee20]{root}
Junho Lee.
\newblock A square root of {H}urwitz numbers.
\newblock {\em Manuscripta Math.}, 162(1-2):99--113, 2020.
\newblock \href {http://dx.doi.org/10.1007/s00229-019-01113-0}
  {\path{doi:10.1007/s00229-019-01113-0}}.

\bibitem[LY21]{Liu}
Xiaobo Liu and Chenglang Yang.
\newblock {Schur Q-Polynomials and Kontsevich-Witten Tau Function}, 2021.
\newblock \href {http://arxiv.org/abs/2103.14318} {\path{arXiv:2103.14318}}.

\bibitem[LY22]{LC}
Xiaobo Liu and Chenglang Yang.
\newblock {Q-polynomial expansion for {B}r\'{e}zin-{G}ross-{W}itten
  tau-function}.
\newblock {\em Adv. Math.}, 404:Paper No. 108456, 28, 2022.
\newblock \href {http://dx.doi.org/10.1016/j.aim.2022.108456}
  {\path{doi:10.1016/j.aim.2022.108456}}.

\bibitem[Mac95]{Mac}
I.~G. Macdonald.
\newblock {\em Symmetric functions and {H}all polynomials}.
\newblock Oxford Mathematical Monographs. The Clarendon Press, Oxford
  University Press, New York, second edition, 1995.
\newblock With contributions by A. Zelevinsky, Oxford Science Publications.

\bibitem[MM]{MMQ}
Andrei Mironov and Alexei Morozov.
\newblock {Superintegrability of Kontsevich matrix model}.
\newblock {\em The European Physical Journal C}, (3):270.
\newblock \href {http://dx.doi.org/10.1140/epjc/s10052-021-09030-x}
  {\path{doi:10.1140/epjc/s10052-021-09030-x}}.

\bibitem[MMN]{MMNQ}
A.~Mironov, A.~Morozov, and S.~Natanzon.
\newblock Cut-and-join structure and integrability for spin {H}urwitz numbers.
\newblock {\em The European Physical Journal C}, (2):97.
\newblock \href {http://dx.doi.org/10.1140/epjc/s10052-020-7650-2}
  {\path{doi:10.1140/epjc/s10052-020-7650-2}}.

\bibitem[MMNO21]{MMNO}
A.~D. Mironov, A.~Morozov, S.~M. Natanzon, and A.~Yu. Orlov.
\newblock Around spin {H}urwitz numbers.
\newblock {\em Lett. Math. Phys.}, 111(5):Paper No. 124, 39, 2021.
\newblock \href {http://dx.doi.org/10.1007/s11005-021-01457-3}
  {\path{doi:10.1007/s11005-021-01457-3}}.

\bibitem[MMS96]{MMS}
A.~Mironov, A.~Morozov, and G.~W. Semenoff.
\newblock Unitary matrix integrals in the framework of the generalized
  {K}ontsevich model.
\newblock {\em Internat. J. Modern Phys. A}, 11(28):5031--5080, 1996.
\newblock \href {http://dx.doi.org/10.1142/S0217751X96002339}
  {\path{doi:10.1142/S0217751X96002339}}.

\bibitem[Nor17]{Norb}
Paul Norbury.
\newblock {A new cohomology class on the moduli space of curves}, 2017.
\newblock \href {http://arxiv.org/abs/1712.03662} {\path{arXiv:1712.03662}}.

\bibitem[Nor20]{N2}
Paul Norbury.
\newblock {Enumerative geometry via the moduli space of super Riemann
  surfaces}, 2020.
\newblock \href {http://arxiv.org/abs/2005.04378} {\path{arXiv:2005.04378}}.

\bibitem[Orl03]{OBKP}
A.~Yu. Orlov.
\newblock Hypergeometric functions associated with {S}chur {$Q$}-polynomials,
  and the {BKP} equation.
\newblock {\em Teoret. Mat. Fiz.}, 137(2):253--270, 2003.
\newblock \href {http://dx.doi.org/10.1023/A:1027370004436}
  {\path{doi:10.1023/A:1027370004436}}.

\bibitem[Sch11]{SchQ}
J.~Schur.
\newblock \"{U}ber die {D}arstellung der symmetrischen und der alternierenden
  {G}ruppe durch gebrochene lineare {S}ubstitutionen.
\newblock {\em J. Reine Angew. Math.}, 139:155--250, 1911.
\newblock \href {http://dx.doi.org/10.1515/crll.1911.139.155}
  {\path{doi:10.1515/crll.1911.139.155}}.

\bibitem[SS83]{Sato}
Mikio Sato and Yasuko Sato.
\newblock Soliton equations as dynamical systems on infinite-dimensional
  {G}rassmann manifold.
\newblock In {\em Nonlinear partial differential equations in applied science
  ({T}okyo, 1982)}, volume~81 of {\em North-Holland Math. Stud.}, pages
  259--271. North-Holland, Amsterdam, 1983.

\bibitem[vdL95]{vdLASM}
Johan van~de Leur.
\newblock The {A}dler-{S}hiota-van {M}oerbeke formula for the {BKP} hierarchy.
\newblock {\em J. Math. Phys.}, 36(9):4940--4951, 1995.
\newblock \href {http://dx.doi.org/10.1063/1.531352}
  {\path{doi:10.1063/1.531352}}.

\bibitem[Wit91]{Wit91}
Edward Witten.
\newblock Two-dimensional gravity and intersection theory on moduli space.
\newblock In {\em Surveys in differential geometry ({C}ambridge, {MA}, 1990)},
  pages 243--310. Lehigh Univ., Bethlehem, PA, 1991.

\bibitem[You89]{You}
Yuching You.
\newblock Polynomial solutions of the {BKP} hierarchy and projective
  representations of symmetric groups.
\newblock In {\em Infinite-dimensional {L}ie algebras and groups
  ({L}uminy-{M}arseille, 1988)}, volume~7 of {\em Adv. Ser. Math. Phys.}, pages
  449--464. World Sci. Publ., Teaneck, NJ, 1989.

\bibitem[Zog15]{Zog}
Peter Zograf.
\newblock Enumeration of {G}rothendieck's dessins and {KP} hierarchy.
\newblock {\em Int. Math. Res. Not. IMRN}, (24):13533--13544, 2015.
\newblock \href {http://dx.doi.org/10.1093/imrn/rnv077}
  {\path{doi:10.1093/imrn/rnv077}}.

\end{thebibliography}

\end{document}